\documentclass[12pt]{iopart}
\usepackage{siunitx}
\newcommand{\dd}{\,\mathrm{d}}

\usepackage{bm}
\usepackage{iopams}

\usepackage{braket}
\usepackage{color}
\usepackage{setstack}
\usepackage{citesort}
\usepackage{cite}
\usepackage{longtable}
\usepackage{caption}
\usepackage{multirow}
\usepackage{graphicx}
\usepackage{ntheorem}
\usepackage{stackrel}
\usepackage{subfigure}
\usepackage{hhline}
\usepackage{epstopdf}
\usepackage[fleqn]{cases}

\usepackage{iopams} 
\newtheorem{proposition}{Proposition}
\newtheorem{lemma}{Lemma}

\newenvironment{proof}[1][Proof]{%
  \par\noindent\textbf{#1.} \hspace*{1ex}%
 \itshape 
}{\hfill$\square$\par}
\newtheorem{theorem}{Theorem}

\begin{document}
\title[Spectral moments of Bures-Hall ensemble]{Spectral moments of Bures-Hall ensemble and applications to entanglement entropy}
\author{Linfeng Wei$^1$, Youyi Huang$^2$, and Lu Wei$^1$}
\address{$^1$Department of Computer Science, Texas Tech University, Texas 79409, USA}
\address{$^2$University of Central Missouri, Warrensburg, Missouri 64093, USA}
\ead{$^1$\{linwei,luwei\}@ttu.edu; $^2$yhuang@ucmo.edu}
\vspace{10pt}
\begin{indented}
\item[]January 2026
\end{indented}

\begin{abstract}
We study spectral moments of the Bures-Hall random matrices ensemble. The main result establishes a recurrence relation for the $k$-th spectral moment valid for a real-valued $k$, in contrast to prevailing results in the literature of different ensembles of assuming an integer $k$. The key to establish the recurrence relation is the obtained Christoffel-Darboux formulas of correlation kernels of the ensemble that avoid tedious summations. As an application of our spectral moment results, we re-derive the formulas of average von Neumann entropy and quantum purity of Bures-Hall ensemble conjectured by Ayana Sarkar and Santosh Kumar. This work is dedicated to the memory of Santosh Kumar.
\end{abstract}

\vspace{2pc}
\noindent{\it Keywords}: Bures-Hall ensemble, quantum entanglement, quantum purity, random matrix theory, spectral moments, von Neumann entropy

\maketitle
\section{Introduction}
For an ensemble $X$ of random matrices of dimension $m$, the $k$-th spectral moment is defined as 
\begin{equation}\label{eq:specmoment}
    \mathbb{E}\left[ \mathrm{Tr}\! \left(X^k\right)\right],
\end{equation}
which has been well studied for classical ensembles including Gaussian, Laguerre and Jacobi unitary ensembles~\cite{1986InMat..85..457H,Haagerup2003RandomMW,Ledoux2004,Ledoux2009ARF,Dubrovin_2017,Cunden_2019,Gisonni_2020,Gisonni_2021,huang2025cumulantstructuresentanglemententropy}, as well as non-Hermitian ensembles~\cite{BYUN2024103526,byun2024qdeformedgaussianunitaryensemble,ByunForrester24realGin,akemann2025spectralmomentscomplexsymplectic,byun2025spectralanalysisqdeformedunitary}. In particular, exact recurrence relations for spectral moments of Gaussian orthogonal ensemble, Gaussian unitary ensemble, and Laguerre unitary ensembles have been studied and obtained in~\cite{1986InMat..85..457H,Haagerup2003RandomMW,Ledoux2004,Ledoux2009ARF,huang2025cumulantstructuresentanglemententropy}. It is worth mentioning that, except for~\cite{Cunden_2019,huang2025cumulantstructuresentanglemententropy}, most of the existing works focus on the spectral moments of an integer order $k$. In this work, we study spectral moments of the Bures-Hall ensemble (\ref{eq:BH}) described below, and deduce a three-term recurrence relation (\ref{specrecrel}) for the spectral moments in terms of a real order $k$. 

Entanglement statistics and related topics have been studied over various ensembles. This includes various impactful contribution by Santosh Kumar~\cite{Kumar2011,Kumar_2015JPA,Kumar_2015PRE,Kumar_2017,Kumar_2019,Sarkar19,Forrester_2019,PhysRevA.102.012405,PhysRevA.103.032423,PhysRevA.104.022438,Forrester2022DifferentialRF,Sarkar2023EntanglementSS,forrester2023computable,LAHA2024129591}. In particular, Sarkar and Kumar conjectured an exact formulae for the mean of von Neumann entropy and quantum purity over Bures-Hall ensemble~\cite{Sarkar19,forrester2025memoriamaspectssantoshkumars}, which was proved using a summation method in~\cite{Wei20BHA}. Further results on variance and skewness of von Neumann entropy and moments of purity over Bures-Hall ensemble have been obtained with the summation method in~\cite{Wei20BH,LW21,wei2025skewnessvonneumannentropy}. Similarly, the moments/cumulants of entanglement statistics over Hilbert-Schmidt ensemble and fermionic Gaussian ensemble that correspond to Laguerre unitary ensemble and Jacobi unitary ensemble have been studied in~\cite{Page93,Foong94,Ruiz95,HLW06,VPO16,Wei17,Wei20,HWC21,Wei_2022,huang2025cumulantstructuresentanglemententropy,forrester2026integrabilityenabledcomputationsrelating} and~\cite{LRV20,LRV21,BHK21,HW22,huang_entropy_2023,HWISIT}, respectively. The existing work and methods treat the calculation of cumulant of linear statistics over these ensembles individually for each order. Moreover, in terms of calculating the higher-order cumulants of von Neumann entropy over the above three ensembles using the existing summation method~\cite{Wei20BH,wei2025skewnessvonneumannentropy,Wei17,Wei20,HWC21,HW22,huang_entropy_2023,HWISIT}, the simplification of nested summation has become an increasingly tedious and case-by-case task. However, in~\cite{huang2025cumulantstructuresentanglemententropy}, a natural and systematic way to calculate higher-order cumulants of entanglement entropy in terms of lower order cumulants was proposed. The key to their approach for Laguerre unitary ensemble is a recurrence relation for the spectral moments in terms of $k\in \mathbb{R}$, and then taking derivative with respect to $k$ for applications to entanglement entropy. In this work, we adopt such an approach to the Bures-Hall ensemble.

The strategy to calculate the moments/cumulants of linear statistics over Bures-Hall ensemble is to utilize correlation kernels of the unconstrained Bures-Hall ensemble~(\ref{eq:BHu}), which are naturally related to the correlation kernels of Cauchy-Laguerre biorthogonal ensemble~\cite{FK16,Bertola14}. A key to the calculation of spectral moments, especially their recurrence relation in terms of real order $k$ is the Christoffel-Darboux formulae of the correlation kernels~\cite{huang2025cumulantstructuresentanglemententropy}, which enables a summation-free representation of the kernels and consequential summation-free calculations. We will derive the Christoffel-Darboux formulas for the correlation kernels of Cauchy-Laguerre biorthogonal ensemble, and utilize such results to derive the recurrence relation for spectral moments over the unconstrained Bures-Hall ensemble~(\ref{eq:BHu}). As an application, we reproduce the results on the mean of entanglement entropy and quantum purity~\cite{Sarkar19,forrester2025memoriamaspectssantoshkumars,Wei20BHA} via the recurrence relation. It is expected that higher-order cumulants of entanglement entropy of Bures-Hall ensemble can also be obtained via higher-order correlators of spectral moments~\cite{Dubrovin_2017,Gisonni_2020,Gisonni_2021}, which may be addressed in a future work.

The bipartite system of the Bures-Hall ensemble~\cite{Hall98,Zyczkowski01,Sommers03,Sommers04,Osipov10,Borot12,FK16,Bertola14,Sarkar19,Wei20BHA,Wei20BH,LW21,wei2025skewnessvonneumannentropy,wei2025averagerelativeentropyrandom} is formulated as follows. For two Hilbert spaces $\mathcal{H}_{A}$ of dimension $m$ and $\mathcal{H}_{B}$ of dimension $n$, the composite system $\mathcal{H}_{A+B}$ is given by the tensor product of the subsystems $\mathcal{H}_{A+B}=\mathcal{H}_{A}\otimes\mathcal{H}_{B}$. A random pure state of the composite system $\mathcal{H}_{A+B}$ is defined as a linear combination of the complete basis $\big\{\Ket{i^{A}}\big\}$ and $\big\{\Ket{j^{B}}\big\}$ of $\mathcal{H}_{A}$ and $\mathcal{H}_{B}$ with random coefficients $z_{i,j}$, 
\begin{equation}
\Ket{\psi}=\sum_{i=1}^{m}\sum_{j=1}^{n}z_{i,j}\Ket{i^{A}}\otimes\Ket{j^{B}},  
\end{equation}
where each $z_{i,j}$ follows the standard complex Gaussian distribution of zero mean and unit variance with the probability constraint $\sum_{i,j}|z_{i,j}|^2=1$. As in~\cite{Sarkar19}, we consider a superposition of the above state as
\begin{equation}\label{eq:SB}
\Ket{\varphi}=\Ket{\psi}+\left(\mathbf{U}\otimes\mathbf{I}_{n}\right)\Ket{\psi},
\end{equation}
where $\mathbf{U}$ is an $m\times m$ unitary random matrix with the measure proportional to $\det^{2\alpha+1}\left(\mathbf{I}_{m}+\mathbf{U}\right)$, where the parameter $\alpha$ takes half-integer values as
\begin{equation}\label{eq:aBH}
\alpha=n-m-\frac{1}{2}.
\end{equation} The corresponding density matrix of the pure state~(\ref{eq:SB}) is
$\rho=\Ket{\varphi}\Bra{\varphi}$,
which follows the natural probability constraint $\tr(\rho)=1$.
Without loss of generality, we assume that $m\leq n$, i.e., the dimension of subsystem $A$ is no more than the dimension of subsystem $B$. The reduced density matrix $\rho_{A}$ of subsystem $A$ is computed by partial tracing of the full density matrix over the other subsystem $B$ as $\rho_{A}=\tr_{B}(\rho)$.
The resulting density of eigenvalues of $\rho_{A}$, $\lambda_{i}\in[0,1]$, $i=1,\dots,m$, is the Bures-Hall measure~\cite{Hall98,Zyczkowski01,Sommers03,Sarkar19}
\begin{equation}\label{eq:BH}
f\left(\bm{\lambda}\right)=\frac{1}{C}~\delta\left(1-\sum_{i=1}^{m}\lambda_{i}\right)\prod_{1\leq i<j\leq m}\frac{\left(\lambda_{i}-\lambda_{j}\right)^{2}}{\lambda_{i}+\lambda_{j}}\prod_{i=1}^{m}\lambda_{i}^{\alpha},
\end{equation}
where the constant $C$ is
\begin{equation}\label{eq:cBH}
C=\frac{2^{-m(2\alpha+m)}\pi^{m/2}}{\Gamma\left(m(2\alpha+m+1)/2\right)}\prod_{i=1}^{m}\frac{\Gamma(i+1)\Gamma(2\alpha+i+1)}{\Gamma(\alpha+i+1/2)}
\end{equation}
with $\delta(x)$ denoting the Dirac delta function and $\Gamma(x)$ being the Gamma function.

To calculate the moments/cumulants of linear statistics over the Bures-Hall ensemble~(\ref{eq:BH}), a standard way is to perform the calculation over the unconstrained Bures-Hall ensemble~\cite{Sarkar19,Wei20BHA,Osipov10,Wei20BH,LW21,wei2025skewnessvonneumannentropy}, whose density function reads 
\begin{equation}\label{eq:BHu}
h\left(\bm{x}\right)=\frac{1}{C'}\prod_{1\leq i<j\leq m}\frac{\left(x_{i}-x_{j}\right)^{2}}{x_{i}+x_{j}}\prod_{i=1}^{m}x_{i}^{\alpha}\e^{-x_{i}},
\end{equation}
where $x_{i}\in[0,\infty)$, $i=1,\dots,m$, and the constant $C'$ is 
\begin{eqnarray}
    C'=C~\Gamma\left(d\right)
\end{eqnarray} with 
\begin{equation}\label{eq:d}
d=\frac{1}{2}m\left(2\alpha+m+1\right). 
\end{equation}
Such unconstrained ensemble is related to the original Bures-Hall ensemble~(\ref{eq:BH}) by the factorization~\cite{Wei20BHA,Wei20BH,LW21,wei2025skewnessvonneumannentropy}
\begin{equation}\label{eq:g2ft}
h(\bm{x})\prod_{i=1}^{m}\dd x_{i}=f(\bm{\lambda})g(\theta)\dd\theta\prod_{i=1}^{m}\dd\lambda_{i},
\end{equation}
where 
\begin{equation}
g(\theta)=\frac{1}{\Gamma\left(d\right)}\e^{-\theta}\theta^{d-1}
\end{equation}is the density of trace of the unconstrained ensemble 
\begin{equation}\label{eq:tr}
\theta=\sum_{i=1}^{m}x_{i}~~~~\theta\in[0,\infty).
\end{equation}
The above factorization leads to a moment conversion between the moments of linear statistics over the unconstrained Bures-Hall ensemble~(\ref{eq:BHu}) and those over the Bures-Hall ensemble~(\ref{eq:BH})~\cite{Wei20BHA,Wei20BH,LW21,wei2025skewnessvonneumannentropy}, which vary between different linear statistics. In particular, we consider a family of linear statistics over the unconstrained Bures-Hall ensemble~(\ref{eq:BHu}),
\begin{equation}
    R_k=\sum_{i=1}^{m}x_i^k
\end{equation}
for $k \in \mathbb{R}$. Clearly, the first cumulant of the linear statistics $R_k$ is the $k$-th spectral moment over the unconstrained Bures-Hall ensemble~(\ref{eq:BHu}),
\begin{eqnarray}\label{eq:Rk}
    \kappa(R_k) =  \mathbb{E}\left[ \sum_{i=1}^{m} x_i^k\right] .
\end{eqnarray}
Also, we define a family of linear statistics over the unconstrained Bures-Hall ensemble~(\ref{eq:BHu}),
\begin{equation}\label{eq:Tk}
    T_k=\sum_{i=1}^{m}x_i^k\ln x_i
\end{equation}
for $k \in \mathbb{R}$, which will be useful in the application to the entanglement entropy.

The rest of the paper is organized as follows. In Section~\ref{sec:Theory}, we derive relevant technical results on Cauchy-Laguerre biorthogonal kernels and deduce their Christoffel-Darboux formulas. We then derive a recurrence relation for spectral moments summarized in Theorem~\ref{Theorem:Rec}, which is the main result of this work. Section~\ref{App} is devoted to applications to von Neumann entropy and quantum purity. Coefficients of various obtained formulas are in~\ref{appendix:A}.

\section{Cauchy-Laguerre biorthogonal kernels and recurrence relation for spectral moments}\label{sec:Theory}
\subsection{Christoffel-Darboux formulas for Cauchy-Laguerre biorthogonal kernels}\label{sec:cd}
The $k$-point ($k\leq m$) correlation function $\rho_{k}(x_{1},\dots,x_{k})$ of the unconstrained Bures-Hall ensemble is known to follow a Pfaffian point process of a $2k\times2k$ antisymmetric matrix~\cite{FK16}, where the corresponding correlation kernels are related to the Cauchy-Laguerre biorthogonal kernels~\cite{Bertola14}. In this section, we derive some technical results of the Cauchy-Laguerre biorthogonal kernels and their polynomials. We will first state known results on the structure relations and four-term expressions of Cauchy-Laguerre biorthogonal polynomials in Lemma \ref{lemma:strucrel} and Lemma \ref{lemma:fourtermrec}, respectively, which are used in the derivation and simplification of derivative of Cauchy-Laguerre biorthogonal kernels. We then derive the key Christoffel-Darboux formulas for the Cauchy-Laguerre biorthogonal kernels in Proposition \ref{Prop:CD}. In Proposition \ref{Prop:DerivK}, we adopt the Christoffel-Darboux formulas to derive the derivative of one-point correlation kernel with simple representations. The summation-free representation of derivative of one-point kernel enables us to perform the integration by parts to derive the recurrence relation (\ref{specrecrel}) for spectral moments (\ref{eq:Rk}), which is presented in Theorem \ref{Theorem:Rec}.
 
The computation of moments/cumulants of linear statistics over the (unconstrained) Bures-Hall ensemble can be performed over the kernels of the Cauchy-Laguerre ensemble, where the four correlation kernels~\cite{FK16, LW21} are $K_{00}(x,y)$, $K_{01}(x,y)$, $K_{10}(x,y)$, and $K_{11}(x,y)$, which admit finite summation representation~\cite{FK16,Bertola14} as below.
\begin{subequations}\label{kernels}
\begin{eqnarray}
    K_{00}(x,y)&=&\sum _{k=0}^{m-1} \frac{1}{h_k}p_k(x) q_k(y) \label{eq:00}\\
    K_{01}(x,y)&=&-x^{\alpha }e^{-x}\sum _{k=0}^{m-1}\frac{1}{h_k} p_k(y) Q_k(-x) \label{eq:01}\\
   K_{10}(x,y)&=&-y^{\alpha +1}e^{-y}\sum _{k=0}^{m-1} \frac{1}{h_k}P_k(-y) q_k(x) \label{eq:10}\\
    K_{11}(x,y)&=&x^{\alpha }y^{\alpha +1} e^{-x-y}  \sum _{k=0}^{m-1}\frac{1}{h_k} P_k(-y) Q_k(-x)-W(x,y). \label{eq:11}
    \end{eqnarray}
    \end{subequations}
    Here, $p_{k}(x)$ and $q_{k}(y)$ are Cauchy-Laguerre biorthogonal polynomials,
\begin{equation}\label{eq:oc}
\int_{0}^{\infty}\!\!\int_{0}^{\infty}p_{k}(x)q_{l}(y)W(x,y)\dd x\dd y=h_k\delta_{k,l}
\end{equation}
with the normalizing constant $h_k$ and weight function 
\begin{equation}
    W(x,y)= \frac{x^{\alpha } y^{\alpha +1} e^{-x-y}}{x+y}.
\end{equation}
Also, the Cauchy transforms of $p_{k}(x)$ and $q_{k}(y)$ are respectively
\begin{subequations}
\begin{eqnarray}\label{eq:PQ}
P_{k}(x)&=&\int_{0}^{\infty}\frac{v^{\alpha}\e^{-v}}{x-v}p_{k}(v)\dd v \\
Q_{k}(y)&=&\int_{0}^{\infty}\frac{w^{\alpha+1}\e^{-w}}{y-w}q_{k}(w)\dd w.
\end{eqnarray}
\end{subequations}
Their representations via Meijer G-functions are also known~\cite{Bertola10,Bertola14,Forrester}
\begin{subequations}\label{eq:kerM}
\begin{eqnarray}
p_{k}(x)&=&\sqrt{2}(-1)^{k}G_{2,3}^{1,1}\left(\!\begin{array}{c}-2\alpha-1-k;k+1\\0;-\alpha,-2\alpha-1\end{array}\Big|x\Big.\right)\\
q_{k}(x)&=&\sqrt{2}(-1)^{k}(\alpha+k+1)G_{2,3}^{1,1}\left(\!\begin{array}{c}-2\alpha-1-k;k+1\\0;-\alpha-1,-2\alpha-1\end{array}\Big|x\Big.\right)\\
P_{k}(x)&=&\sqrt{2}(-1)^{k+1}\e^{-x}G_{2,3}^{2,1}\left(\!\begin{array}{c}-\alpha-k-1;\alpha+k+1\\0,\alpha;-\alpha-1\end{array}\Big|-\!x\Big.\right)\\
Q_{k}(x)&=&\sqrt{2}(-1)^{k+1}(\alpha+k+1)\e^{-x} G_{2,3}^{2,1}\left(\!\begin{array}{c}-\alpha-k;\alpha+k+2\\0,\alpha+1;-\alpha\end{array}\Big|-\!x\Big.\right). 
\end{eqnarray}
\end{subequations}
Moreover, explicit representation of the Cauchy-Laguerre biorthogonal polynomials are available from definition of Meijer G-functions~\cite{Prudnikov86},
\begin{subequations}\label{pqpoly}
\begin{eqnarray}
\fl    p_m(x)&=&  \sqrt{2} (-1)^m \sum _{j=0}^m \frac{(-1)^j  \Gamma (2 \alpha +j+m+2)x^j}{j! \Gamma (\alpha +j+1) \Gamma (2 \alpha +j+2) \Gamma (-j+m+1)}\label{eq:pexplicit}\\
\fl  q_m(x)&=&  \sqrt{2} (-1)^m (\alpha +m+1)\sum _{j=0}^m \frac{(-1)^j  \Gamma (2 \alpha +j+m+2)x^j}{j! \Gamma (\alpha +j+2) \Gamma (2 \alpha +j+2) \Gamma (-j+m+1)}.\label{eq:qexplicit}
\end{eqnarray} 
\end{subequations}
Another important ingredient in the results below for Cauchy-Laguerre biorthogonal kernels is the highest order coefficient of the biorthogonal polynomials, which is the same for both $p_m(x)$ and $q_m(x)$, and we denote as $S_m$, i.e.,
\begin{eqnarray}
    p_m(x)&=&S_mx^m+\mbox{lower-order terms}\\
    q_m(x)&=&S_mx^m+\mbox{lower-order terms},
\end{eqnarray}
where 
\begin{eqnarray}
    S_m=\frac{2 \sqrt{2} (2 \alpha +2 m+1) \Gamma (2 \alpha+2 m )}{m (2 \alpha +m) (2 \alpha +m+1) \Gamma (m) \Gamma (\alpha+m ) \Gamma (2 \alpha+m )}.
\end{eqnarray}

The spectral moments of unconstrained Bures-Hall ensemble can be calculated over the one point density as 
\begin{eqnarray}
\kappa(R_k)&=& m\int_{0}^{\infty}x^k~h_{1}(x)\dd x,\label{eq:Rk}
\end{eqnarray}
where the density function $h_1(x)$ is~\cite{LW21,FK16}
\begin{eqnarray}\label{h1}
h_{1}(x)&=&\frac{1}{2m}\left(K_{01}(x,x)+K_{10}(x,x)\right).\label{eq:h1}
\end{eqnarray}

We first state some known results from \cite{witte2022gapprobabilitiesbureshallensemble}.
\begin{subequations}\label{lemmasquare}
   \begin{eqnarray}
        &&\int_0^{\infty}\int_0^{\infty} \frac{x^\alpha y^{\alpha+1} e^{-x-y} }{(x+y)^2} x p_m(x) q_n(y)\dd x \dd y\nonumber\\
        &=&\left\{\begin{array}{rr}
-\frac{S_m}{S_{m+1}} ,& n=m+1\\
-\frac{S_{m-1}}{S_m}+\frac{S_m}{S_{m+1}} ,& n=m \\
\frac{S_{m-1}}{S_m} ,& n=m-1\\
0 ,& \mbox{otherwise} \\
\end{array}\right.
    \end{eqnarray}
    \begin{eqnarray}
&&\int_0^{\infty}\int_0^{\infty} \frac{x^\alpha y^{\alpha+1} e^{-x-y} }{(x+y)^2}  p_m(x) y q_n(y)\dd x \dd y\nonumber\\
        &=&\left\{\begin{array}{rr}
\frac{S_m}{S_{m+1}} ,& n=m+1\\
\frac{S_{m-1}}{S_m}-\frac{S_m}{S_{m+1}}+1 ,& n=m \\
-\frac{S_{m-1}}{S_m} ,& n=m-1\\
0, & \mbox{otherwise} \\
\end{array}\right.
    \end{eqnarray} 
\end{subequations}
The above results can be derived using integration by parts as
\begin{eqnarray}
&&\int_0^{\infty}\int_0^{\infty} \frac{x^\alpha y^{\alpha+1} e^{-x-y} }{(x+y)^2} x p_m(x) q_n(y)\dd x \dd y\nonumber\\
&=&\int_{0}^{\infty}\int_{0}^{\infty}x\left(\frac{\dd}{\dd x}p_m(x)\right)q_n(y)\frac{x^\alpha y^{\alpha+1} e^{-x-y}}{x+y} \dd x\dd y+(\alpha+1)\delta_{m,n}\nonumber\\
&&-\int_{0}^{\infty}\int_{0}^{\infty}x p_m(x)q_n(y)\frac{x^\alpha y^{\alpha+1} e^{-x-y}}{x+y} \dd x\dd y,
\end{eqnarray}
and 
\begin{eqnarray}
&&\int_0^{\infty}\int_0^{\infty} \frac{x^\alpha y^{\alpha+1} e^{-x-y} }{(x+y)^2} x p_m(x) y q_n(y)\dd x \dd y\nonumber\\
&=&\int_{0}^{\infty}\int_{0}^{\infty}p_m(x)y\left(\frac{\dd}{\dd y}q_n(y)\right)\frac{x^\alpha y^{\alpha+1} e^{-x-y}}{x+y} \dd x\dd y+(\alpha+2)\delta_{m,n}\nonumber\\
&&-\int_{0}^{\infty}\int_{0}^{\infty} p_m(x)y q_n(y)\frac{x^\alpha y^{\alpha+1} e^{-x-y}}{x+y} \dd x\dd y.
\end{eqnarray}
Structure relations for $p_m(x)$, $q_m(x)$, $P_m(x)$, $Q_m(x)$ are also obtained in~\cite{witte2022gapprobabilitiesbureshallensemble} as follows. 
\begin{lemma}\label{lemma:strucrel}
\begin{subequations}\label{eq:strucrel}
    \begin{eqnarray}\label{eq:strucp}
  \fl ~~~~~~     x\frac{\dd}{\dd x}p_m(x) &=&\frac{S_{m-1}}{S_m}p_{m-1}(x)-\left(\alpha +1-\frac{S_m}{S_{m+1}}+\frac{S_{m-1}}{S_m}\right)p_m(x)\nonumber\\ \fl &&
        +x p_m(x)-\frac{S_m }{S_{m+1}}p_{m+1}(x)
    \end{eqnarray}
        \begin{eqnarray}\label{eq:strucq}
\fl ~~~~~~ x\frac{\dd}{\dd x}q_m(x) &=&\frac{S_{m-1}}{S_m}q_{m-1}(x)-\left(\alpha +1+\frac{S_m}{S_{m+1}}-\frac{S_{m-1}}{S_m}\right)q_m(x)\nonumber\\ \fl &&
       +x q_m(x)-\frac{S_m }{S_{m+1}}q_{m+1}(x)
    \end{eqnarray}
    
\begin{eqnarray}\label{eq:strucP}
  \fl  ~~~~~~    x\frac{\dd}{\dd x}P_m(x) &=& \frac{S_{m-1}}{S_m}P_{m-1}(x)
            -\left(1-\frac{S_m}{S_{m+1}}+\frac{S_{m-1}}{S_m}\right)P_m(x)-\frac{S_m }{S_{m+1}}P_{m+1}(x)
    \end{eqnarray}
    
\begin{eqnarray}\label{eq:strucQ}
 \fl  ~~~~~~     x\frac{\dd}{\dd x} Q_m(x)&=& \frac{S_{m-1}}{S_m}Q_{m-1}(x)-\left(\frac{S_m}{S_{m+1}}-\frac{S_{m-1}}{S_m}\right)Q_m(x)-\frac{S_m }{S_{m+1}}Q_{m+1}(x).
    \end{eqnarray}       
\end{subequations}
\end{lemma}

\begin{proof}
    Let $a_{m,j}$, $b_{m,j}$ be two set of numbers such that 
    \begin{subequations}
        \begin{equation}
            x\frac{\dd}{\dd x}p_m(x)= \sum_{j=0}^{m}a_{m,j}p_j(x)
        \end{equation}
         \begin{equation}
            x\frac{\dd}{\dd x}q_m(x)= \sum_{j=0}^{m}b_{m,j}q_j(x).
        \end{equation}
    \end{subequations}
     Similar to the proof of Lemma \ref{prop:cjkdkj} below, we have      
      \begin{subequations}
     \begin{equation}
a_{m,j}=\int_{0}^{\infty}\int_{0}^{\infty}x\left(\frac{\dd}{\dd x}p_m(x)\right)q_j(y)\frac{x^\alpha y^{\alpha+1} e^{-x-y}}{x+y} \dd x\dd y
     \end{equation}
     \begin{equation}
b_{m,j}=\int_{0}^{\infty}\int_{0}^{\infty}y\left(\frac{\dd}{\dd y}q_m(y)\right)p_j(x)\frac{x^\alpha y^{\alpha+1} e^{-x-y}}{x+y} \dd x\dd y.
     \end{equation}
\end{subequations}
On the other hand, do integration by parts and we have 
  \begin{eqnarray}
     &&\!\!\!\!\!\!\!\! \int_{0}^{\infty}\int_{0}^{\infty}x\left(\frac{\dd}{\dd x}p_m(x)\right)q_j(y)\frac{x^\alpha y^{\alpha+1} e^{-x-y}}{x+y} \dd x\dd y\nonumber\\
      &\!\!\!\!\!\!\!\!=&\!\! -\!\!\int_{0}^{\infty}\!\!\!\!\int_{0}^{\infty}\!\!p_m(x)q_j(y)\left( \frac{ x^\alpha y^{\alpha+2} e^{-x-y}}{(x+y)^2}+\frac{x}{x+y}\frac{\dd}{\dd x}(x^\alpha y^{\alpha+1} e^{-x-y}) \right)\!\dd x\!\dd y.  
  \end{eqnarray}
   Hence we have     
   \begin{equation}
        a_{m,j}=-\int_{0}^{\infty}\int_{0}^{\infty}p_m(x)q_j(y)x^\alpha y^{\alpha+1} e^{-x-y}\left(\frac{y}{(x+y)^2}+\frac{\alpha-x}{x+y}\right),
   \end{equation}
   and similarly 
   \begin{equation}
        b_{m,j}=-\int_{0}^{\infty}\int_{0}^{\infty}p_j(x)q_n(y)x^\alpha y^{\alpha+1} e^{-x-y}\left(\frac{x}{(x+y)^2}+\frac{\alpha+1-y}{x+y}\right).
   \end{equation}
   Insert the results (\ref{lemmasquare}) for the $\frac{1}{(x+y)^2}$, and evaluate the $\frac{1}{x+y}$ term with highest order coefficients below (\ref{eq:cmm-1}), (\ref{eq:dmm-1}), we complete the proof of (\ref{eq:strucp}), (\ref{eq:strucq}).
   
Integrate over $x\in[0,\infty)$ on the both sides of (\ref{eq:strucp}) and (\ref{eq:strucq}) multiplied by 
\begin{eqnarray}
      \frac{1}{z-x}x^\alpha e^{-x}
      \end{eqnarray}
      and 
      \begin{eqnarray}
        \frac{1}{z-x}x^{\alpha+1} e^{-x},
\end{eqnarray} respectively, and do integration by parts with 
\begin{equation}
        \frac{\dd}{\dd z} \frac{1}{z-x}=-\frac{\dd}{\dd x} \frac{1}{z-x} 
    \end{equation}
    for the left-hand sides, we have consequently the results (\ref{eq:strucP}) and (\ref{eq:strucQ}), respectively.
    \end{proof}

We also need the following recurrence relations derived in~\cite{witte2022gapprobabilitiesbureshallensemble}.
\begin{lemma}\label{lemma:fourtermrec}
\begin{subequations}\label{eq:fourtermrec}
    \begin{eqnarray}\label{precrel}
        &&x \left(p_{m+1}(x)-p_m(x)\right)\nonumber\\
        &=&r_{m,2}p_{m+2}(x)+r_{m,1}p_{m+1}(x)+r_{m,0}p_{m}(x)+r_{m,-1}p_{m-1}(x)
    \end{eqnarray}
    
\begin{eqnarray}\label{qrecrel}
      &&x\left(\frac{q_{m+1}(x)}{\alpha+m+2} -\frac{q_m(x)}{\alpha+m+1} \right)\nonumber\\
      &=&s_{m, 2} q_{m+2}(x)+s_{m, 1} q_{m+1}(x)+s_{m, 0} q_m(x)+s_{m,-1} q_{m-1}(x)
\end{eqnarray}  

    \begin{eqnarray}\label{Precrel}
        &&x \left(P_{m+1}(x)-P_m(x)\right)\nonumber\\
        &=&r_{m,2}P_{m+2}(x)+r_{m,1}P_{m+1}(x)+r_{m,0}P_{m}(x)+r_{m,-1}P_{m-1}(x)
    \end{eqnarray}
    
\begin{eqnarray}\label{Qrecrel}
      &&x\left(\frac{Q_{m+1}(x)}{\alpha+m+2} -\frac{Q_m(x)}{\alpha+m+1} \right)\nonumber\\
      &=&s_{m, 2} Q_{m+2}(x)+s_{m, 1} Q_{m+1}(x)+s_{m, 0} Q_m(x)+s_{m,-1} Q_{m-1}(x),
\end{eqnarray}
\end{subequations}

\begin{eqnarray*}
  \fl r_{m, 2}\!&=&\frac{S_{m+1}}{S_{m+2} }=\frac{(m+2) (2 \alpha +m+3)}{ 2(2 \alpha +2 m+5)} \\
  \fl r_{m, 1}\!&=&\frac{4 \alpha ^2+14 (\alpha +1)+6 \alpha  m+m (3 m+13)}{2(2 \alpha +2 m+5)} \\
   \fl r_{m, 0}\!&=&\frac{4 \alpha ^2+3 m^2+6 \alpha  (m+1)+5 m+2}{2(2 \alpha +2 m+1)} \\
  \fl r_{m,-1}\!&=&\frac{S_{m-1}}{S_m}=\frac{ m (2 \alpha +m+1)}{2(2 \alpha +2 m+1)} \\
  \fl s_{m, 2}\!&=&\frac{S_{m+1}}{S_{m+2} (\alpha+m+2)}=\frac{(m+2) (2 \alpha +m+3)}{2  (\alpha +m+2) (2 \alpha +2 m+5)}\\
 \fl s_{m, 1}\!&=&\frac{4 \alpha ^3+2 \alpha ^2 (5 m+11)+\alpha  (m+2) (9 m+17)+(m+1) (m+2) (3 m+8)}{2 (\alpha +m+1) (\alpha +m+2) (2 \alpha +2 m+5)}\\
 \fl s_{m, 0}\!&=&\frac{2 (\alpha +1)^2 (2 \alpha +1)+3 m^3+(9 \alpha +10) m^2+(2 \alpha +3) (5 \alpha +3) m}{2 (\alpha +m+1) (\alpha +m+2) (2 \alpha +2 m+1)}\\
 \fl s_{m,-1}\!&=&\frac{S_{m-1}}{S_m (\alpha+m+1)}=\frac{m (2 \alpha +m+1)}{2  (\alpha +m+1) (2 \alpha +2 m+1)} .
\end{eqnarray*}
\end{lemma}
\begin{proof}
The results~(\ref{precrel}), (\ref{qrecrel}) are available in~\cite{witte2022gapprobabilitiesbureshallensemble} and can be obtained from the results on $c_{j,k}$, $d_{k,j}$ in~(\ref{eq:cjk}), (\ref{eq:dkj}) below. We give proof of (\ref{Precrel}) from (\ref{precrel}), where (\ref{Qrecrel}) from (\ref{qrecrel}) is similar.
    We have from (\ref{precrel}) \begin{eqnarray}
      &&z \left(p_{m+1}(z)-p_m(z)\right)\nonumber\\
      &=&r_{m,2}p_{m+2}(z)+r_{m,1}p_{m+1}(z)+r_{m,0}p_{m}(z)+r_{m,-1}p_{m-1}(z).
\end{eqnarray}
Hence we have 
\begin{eqnarray}\label{eq:pzx}
  &&   z \left(p_{m+1}(z)-p_{m+1}(x)+p_{m+1}(x)-(p_m(z)-p_{m}(x)+p_{m}(x))\right)\nonumber\\
  &&-x \left(p_{m+1}(x)-p_m(x)\right)\nonumber\\
& =&r_{m,2}(p_{m+2}(z)-p_{m+2}(x))+r_{m,1}(p_{m+1}(z)-p_{m+1}(x))\nonumber\\&&+r_{m,0}(p_{m}(z)-p_{m}(x))+r_{m,-1}(p_{m-1}(z)-p_{m-1}(x)).
\end{eqnarray}
 Set\begin{equation}\label{eq:ptilde}
        \tilde{P}_{k-1}(z)= \int_0^{\infty}\frac{p_k(z)-p_k(x)}{z-x}x^\alpha e^{-x}\dd x,
    \end{equation}
    and \begin{equation}\label{eq:f1}
        f_1(z)= \int_0^{\infty}\frac{1}{z-x}x^\alpha e^{-x}\dd x.
    \end{equation}
    Multiplying 
\begin{equation}
    \frac{1}{z-x}x^\alpha e^{-x}
\end{equation}
and integrating over $x\in[0,\infty)$ on both sides of (\ref{eq:pzx}), we have
\begin{eqnarray}
      &&z\tilde{P}_{m}(z)-z\tilde{P}_{m-1}(z)+\int_0^{\infty}p_{m+1}(x)x^\alpha e^{-x} \dd x-\int_0^{\infty}p_{m}(x)x^\alpha e^{-x} \dd x\nonumber\\
      &=&r_{m,2}\tilde{P}_{m+1}(z)+r_{m,1}\tilde{P}_{m}(z)+r_{m,0}\tilde{P}_{m-1}(z)+r_{m,-1}\tilde{P}_{m-2}(z),
\end{eqnarray}
where the last two integrals one the left-hand side cancel with each other from the fact of~(\ref{eq:pint}) below.

We now invoke the result in~\cite{witte2022gapprobabilitiesbureshallensemble} that $P_k(z)$ can be written in terms of $p_k(z)$, $f_1(z)$, $\tilde{P}_{k-1}(z)$ as
\begin{equation}\label{eq:pkptilde}
P_k(z)=f_1(z)p_k(z)-\tilde{P}_{k-1}(z).
\end{equation}
The above leads to
\begin{eqnarray}\label{eq:zPintermid}
    &&z \left(f_1(z)p_{m+1}(z)-f_1(z)p_{m}(z)\right)    - z \left(P_{m+1}(z)-P_{m}(z)\right)\nonumber\\&=&r_{m,2}((f_1(z)p_{m+2}(z)-P_{m+2}(z))+r_{m,1}((f_1(z)p_{m+1}(z)-P_{m+1}(z))\nonumber\\&&+r_{m,0}((f_1(z)p_{m}(z)-P_{m}(z))+r_{m,-1}((f_1(z)p_{m-1}(z)-P_{m-1}(z)).
\end{eqnarray}
First term on left-hand side of (\ref{eq:zPintermid}) cancel with $p_k(z)$ and $f_1(z)$ terms on right-hand side through (\ref{precrel}), which completes the proof of (\ref{Precrel}).
\end{proof}

To avoid summation form of the kernels~(\ref{eq:01}), (\ref{eq:10}), we shall use Christoffel-Darboux form of the Cauchy-Laguerre kernels, which was first obtained in~\cite{witte2022gapprobabilitiesbureshallensemble} via Hessenberg matrices. In Proposition~\ref{Prop:CD} below, we give an elementary proof of the results, for which some intermediate results are stated first.
\begin{lemma}\label{prop:cjkdkj}
Let $c_{j,k}$, $d_{k,j}$ be two sets of numbers such that
    \begin{subequations}
    \begin{eqnarray}
          xp_k(x)&=&\sum_{j=0}^{k+1}c_{j,k}p_j(x) \label{eq:xpx}\\
        xq_j(x)&=&\sum_{k=0}^{j+1}d_{k,j}q_k(x), \label{eq:xqx}
    \end{eqnarray}
      \end{subequations}
      we then have
       \begin{subequations}
       \begin{eqnarray}
            c_{j,k}&=&\int_{0}^{\infty}\int_{0}^{\infty}p_k(x)xq_j(y)\frac{x^\alpha y^{\alpha+1} e^{-x-y} }{x+y}\dd x\dd y \label{eq:cjk}\\
     d_{k,j}&=&\int_{0}^{\infty}\int_{0}^{\infty}p_k(x)yq_j(y)\frac{x^\alpha y^{\alpha+1} e^{-x-y} }{x+y}\dd x\dd y,\label{eq:dkj}
       \end{eqnarray}
 \end{subequations}
 and
\begin{equation}\label{eq:cjkdkj}
c_{j,k}+d_{k,j}=2(\alpha+j+1).
\end{equation}
\end{lemma}
\begin{proof}
    Insert (\ref{eq:xpx}) into right-hand side of (\ref{eq:cjk}) and use biorthogonality (\ref{eq:oc}) to obtain (\ref{eq:cjk}), and similarly for the (\ref{eq:dkj}). For (\ref{eq:cjkdkj}), notice that
   \begin{eqnarray}
    c_{j,k}+d_{k,j}&=&\int_{0}^{\infty}\int_{0}^{\infty}p_k(x)q_j(y)x^\alpha y^{\alpha+1} e^{-x-y} \dd x \dd y\nonumber\\
        &=&\int_{0}^{\infty}p_k(x)x^\alpha e^{-x}\dd x\int_{0}^{\infty}q_j(y)y^{\alpha+1} e^{-y}\dd y,
   \end{eqnarray}
   where by (\ref{eq:pexplicit}), (\ref{eq:qexplicit}) \begin{eqnarray}
  \int_{0}^{\infty}p_k(x)x^\alpha e^{-x}\dd x &=&\sqrt{2}\label{eq:pint}\\
  \int_{0}^{\infty}q_j(y)y^{\alpha+1} e^{-y}\dd y &=&\sqrt{2}(\alpha+j+1).\label{eq:qint}
    \end{eqnarray}
\end{proof}

In particular,
    \begin{equation}\label{eq:cmm-1}
        c_{m,m-1}= \int_{0}^{\infty}\int_{0}^{\infty}p_{m-1}(x)xq_m(y)\frac{x^\alpha y^{\alpha+1} e^{-x-y} }{x+y}\dd x \dd y=\frac{S_{m-1}}{S_m}
    \end{equation}
    \begin{equation}\label{eq:dmm-1}
        d_{m,m-1}= \int_{0}^{\infty}\int_{0}^{\infty}p_m(x)yq_{m-1}(y)\frac{x^\alpha y^{\alpha+1} e^{-x-y} }{x+y}\dd x\dd y=\frac{S_{m-1}}{S_m}.\label{eq:dmm-1}
    \end{equation}
Also, we denote $S_m^{m-1}$ the $m-1$ degree term coefficient of $p_m(x)$, then by equating the $x^m$ term coefficients in (\ref{eq:pexplicit}) we have 
\begin{equation}
   c_{m,m}=\frac{S_m^{m-1}}{S_m}-\frac{S_{m+1}^m}{S_{m+1}}=\alpha+m+1+  \frac{S_{m-1}}{S_m}-\frac{S_m}{S_{m+1}},\label{eq:cmmhighest}
\end{equation}
and hence by (\ref{eq:cjkdkj}) we have 
\begin{equation}
   d_{m,m}=\alpha+m+1-  \frac{S_{m-1}}{S_m}+\frac{S_m}{S_{m+1}}.\label{eq:dmmhighest}
\end{equation}

It is worth mentioning that, while the results in~(\ref{singlesums}) below are consequences of Christoffel-Darboux type formulas with single summations in~\cite{witte2022gapprobabilitiesbureshallensemble}, we derive these results by elementary methods using biorthogonality, and adopt them in the derivation and interpretation of Christoffel-Darboux formulas without any summation of Cauchy-Laguerre biorthogonal kernels.
\begin{lemma}
\begin{subequations}\label{singlesums}
     \begin{eqnarray}
\fl            2\sum _{j=0}^{m-1}  (\alpha +j+1) p_j(x)&=&\frac{S_{m-1}}{S_m}p_{m-1}(x)-\frac{S_m }{S_{m+1}}p_{m+1}(x)+x p_m(x)\nonumber\\
            &&-\left(\alpha +m+1-\frac{S_m}{S_{m+1}}+\frac{S_{m-1}}{S_m}\right)p_m(x) \label{eq:phat}
\end{eqnarray}
   
\begin{eqnarray}
\fl        2(\alpha +m+1) \sum _{k=0}^{m-1} q_k(y)&=& \frac{S_{m-1}}{S_m}q_{m-1}(y)-\frac{S_m }{S_{m+1}}q_{m+1}(y)+y q_m(y)\nonumber\\
         &&- \left(\alpha+m+1 +\frac{S_m}{S_{m+1}}-\frac{S_{m-1}}{S_m}\right)q_m(y) \label{eq:qcheck}
     \end{eqnarray}

\begin{eqnarray}
\fl            2\sum _{j=0}^{m-1}  (\alpha +j+1) P_j(x)+\sqrt{2}&=&\frac{S_{m-1}}{S_m}P_{m-1}(x)-\frac{S_m }{S_{m+1}}P_{m+1}(x)+x P_m(x)\nonumber\\
            &&-\left(\alpha +m+1-\frac{S_m}{S_{m+1}}+\frac{S_{m-1}}{S_m}\right)P_m(x) \label{eq:Phat}
\end{eqnarray}
   
\begin{eqnarray}\label{eq:Qcheck}
\fl 2 (\alpha +m+1) \left(\sum _{k=0}^{m-1} Q_k(y)+\frac{1}{\sqrt{2}}\right)&=&\frac{S_{m-1} }{S_m}Q_{m-1}(y)-\frac{S_m }{S_{m+1}}Q_{m+1}(y)+y Q_m(y)\nonumber \\
&& -\left(\alpha +m+1+\frac{S_m}{S_{m+1}}-\frac{S_{m-1}}{S_m}\right)Q_m(y).
\end{eqnarray}
\end{subequations}   
\end{lemma}

\begin{proof}
    We first give proof of (\ref{eq:phat}), and (\ref{eq:qcheck}) is similar. From (\ref{eq:xpx}) we have   \begin{equation}
        xp_m(x)=\sum_{j=0}^{m+1}c_{j,m}p_j(x).
    \end{equation}
    From (\ref{eq:cjkdkj}) we know that \begin{equation}
        c_{j,m}=2(\alpha +j+1)-d_{m,j}.
    \end{equation}
    But $d_{m,j}$ has non-zero value only when $j\geq m-1$
    i.e. \begin{equation}\label{eq:dterminate}
        d_{m,j}=0 \mbox{ for } j \leq m-2.
    \end{equation}
    Hence we have 
    \begin{equation}
         c_{j,m}=2(\alpha +j+1) \mbox{ for } j \leq m-2.
    \end{equation}
    That is \begin{eqnarray}
         xp_m(x)&=&\sum_{j=0}^{m-2}2(\alpha +j+1) p_j(x)\nonumber\\
         &&+c_{m-1,m}p_{m-1}(x)+c_{m,m}p_{m}(x)+c_{m+1,m}p_{m+1}(x).
    \end{eqnarray}
    Also \begin{eqnarray}
          c_{m-1,m}&=&2(\alpha+m)-d_{m,m-1}\nonumber\\
          &=& 2(\alpha+m)-\frac{S_{m-1}}{S_m} \qquad \mbox{by (\ref{eq:dmm-1})}
    \end{eqnarray}
 \begin{equation}
     c_{m+1,m}=\int_{0}^{\infty}\!\int_{0}^{\infty}p_m(x)xq_{m+1}(y)\frac{x^\alpha y^{\alpha+1} e^{-x-y} }{x+y}dxdy=\frac{S_m}{S_{m+1}}.
      \end{equation} 
Hence we have \begin{eqnarray}
     xp_m(x)&=&\sum_{j=0}^{m-1}2(\alpha +j+1) p_j(x)\nonumber\\
     &&-\frac{S_{m-1}}{S_m}p_{m-1}(x)+c_{m,m}p_{m}(x)+\frac{S_m}{S_{m+1}}p_{m+1}(x). \label{eq:phatcmm}
\end{eqnarray}
 Insert (\ref{eq:cmmhighest}) into (\ref{eq:phatcmm}) and one obtains (\ref{eq:phat}). Similarly one can obtain (\ref{eq:qcheck}).

 We can then derive (\ref{eq:Phat}) and (\ref{eq:Qcheck}) from (\ref{eq:phat}) and (\ref{eq:qcheck}), respectively. We demonstrate the process from (\ref{eq:qcheck}) to (\ref{eq:Qcheck}), while the process from (\ref{eq:phat}) to (\ref{eq:Phat}) is similar.
Set
\begin{eqnarray}\label{eq:Qtilde}
     \tilde{Q}_{k-1}(z) &=& \int_0^{\infty}\frac{q_k(z)-q_k(y)}{z-y}y^{\alpha+1} e^{-y}\dd y \label{eq:qtilde}
     \end{eqnarray}
     and 
 \begin{eqnarray}    
      f_2(z) &=& \int_0^{\infty}\frac{1}{z-y}y^{\alpha+1} e^{-y}\dd y.\label{eq:f2}
\end{eqnarray}
Similar to~(\ref{eq:pkptilde}), we have that 
\begin{eqnarray}\label{eq:qkqtilde}
    Q_k(z)=f_2(z)q_k(z)-\tilde{Q}_{k-1}(z).
\end{eqnarray}
  Multiply   \begin{equation}\label{eq:zy}
    \frac{1}{z-y}y^{\alpha+1} e^{-y}
\end{equation}
on the difference of $z$-copy and $y$-copy of (\ref{eq:qcheck}) and integrate over $y\in[0,\infty)$ with terms rewritten by $ \tilde{Q}_{k-1}(z)$ and $f_2(z)$, we have (\ref{eq:Qcheck}), where we have used (\ref{eq:qint}) for integral term
\begin{eqnarray}
   \int_0^\infty q_m(y) y^{\alpha+1}e^{-y}\dd y,
\end{eqnarray} and the identity~(\ref{eq:qkqtilde}). Similarly one can obtain (\ref{eq:Phat}) from (\ref{eq:phat}) using (\ref{eq:ptilde}), (\ref{eq:f1}) and~(\ref{eq:pkptilde}). We have completed the proof.
\end{proof}
 
As a consequence of the above results, we now derive the Christoffel-Darboux formulas for Cauchy-Laguerre kernels from their biorthogonality.
\begin{proposition}\label{Prop:CD}
\begin{subequations}\label{CDs}
        \begin{eqnarray}
      && (x+y)\sum_{k=0}^{m-1}p_k(x)q_k(y)\nonumber\\
        &=&\frac{1}{2 (\alpha +m+1)}\left(\frac{S_{m-1}}{S_m}p_{m-1}(x)-\frac{S_m }{S_{m+1}}p_{m+1}(x)+x p_m(x)\right.\nonumber\\
            &&\left.-\left(\alpha +m+1-\frac{S_m}{S_{m+1}}+\frac{S_{m-1}}{S_m}\right)p_m(x)\right)\left(\frac{S_{m-1}}{S_m}q_{m-1}(y)\right.\nonumber\\
         &&\left.-\frac{S_m }{S_{m+1}}q_{m+1}(y)+y q_m(y)- \left(\alpha+m+1 +\frac{S_m}{S_{m+1}}-\frac{S_{m-1}}{S_m}\right)q_m(y)\right)\nonumber\\
        &&+\frac{S_{m-1}}{S_m}\left(p_m(x)q_{m-1}(y)+p_{m-1}(x)q_{m}(y)\right)\label{CD}
\end{eqnarray}

    \begin{eqnarray}
      && (x+y)\sum_{k=0}^{m-1}P_k(x)q_k(y)\nonumber\\
        &=&\frac{1}{2 (\alpha +m+1)}\left(\frac{S_{m-1}}{S_m}P_{m-1}(x)-\frac{S_m }{S_{m+1}}P_{m+1}(x)+x P_m(x)\right.\nonumber\\
            &&\left.-\left(\alpha +m+1-\frac{S_m}{S_{m+1}}+\frac{S_{m-1}}{S_m}\right)P_m(x)\right)\left(\frac{S_{m-1}}{S_m}q_{m-1}(y)\right.\nonumber\\
         &&\left.-\frac{S_m }{S_{m+1}}q_{m+1}(y)+y q_m(y)- \left(\alpha+m+1 +\frac{S_m}{S_{m+1}}-\frac{S_{m-1}}{S_m}\right)q_m(y)\right)\nonumber\\
        &&+\frac{S_{m-1}}{S_m}\left(P_m(x)q_{m-1}(y)+P_{m-1}(x)q_{m}(y)\right)\label{CDPq}
\end{eqnarray}

\begin{eqnarray}
       &&(x+y)\sum_{k=0}^{m-1}p_k(x)Q_k(y)\nonumber\\
        &=&\frac{1}{2 (\alpha +m+1)}\left(\frac{S_{m-1}}{S_m}p_{m-1}(x)-\frac{S_m }{S_{m+1}}p_{m+1}(x)+x p_m(x)\right.\nonumber\\
            &&\left.-\left(\alpha +m+1-\frac{S_m}{S_{m+1}}+\frac{S_{m-1}}{S_m}\right)p_m(x)\right)\left(\frac{S_{m-1}}{S_m}Q_{m-1}(y)\right.\nonumber\\
         &&\left.-\frac{S_m }{S_{m+1}}Q_{m+1}(y)+y Q_m(y)- \left(\alpha+m+1 +\frac{S_m}{S_{m+1}}-\frac{S_{m-1}}{S_m}\right)Q_m(y)\right)\nonumber\\
        &&+\frac{S_{m-1}}{S_m}\left(p_m(x)Q_{m-1}(y)+p_{m-1}(x)Q_{m}(y)\right)\label{CDpQ}
\end{eqnarray}

 \begin{eqnarray}
     &&  (x+y)\sum_{k=0}^{m-1}P_k(x)Q_k(y)\nonumber\\
        &=&-1+\frac{1}{2 (\alpha +m+1)}\left(\frac{S_{m-1}}{S_m}P_{m-1}(x)-\frac{S_m }{S_{m+1}}P_{m+1}(x)+x P_m(x)\right.\nonumber\\
            &&\left.-\left(\alpha +m+1-\frac{S_m}{S_{m+1}}+\frac{S_{m-1}}{S_m}\right)P_m(x)\right)\left(\frac{S_{m-1}}{S_m}Q_{m-1}(y)\right.\nonumber\\
         &&\left.-\frac{S_m }{S_{m+1}}Q_{m+1}(y)+y Q_m(y)- \left(\alpha+m+1 +\frac{S_m}{S_{m+1}}-\frac{S_{m-1}}{S_m}\right)Q_m(y)\right)\nonumber\\
        &&+\frac{S_{m-1}}{S_m}\left(P_m(x)Q_{m-1}(y)+P_{m-1}(x)Q_{m}(y)\right).\label{CDPQ}
\end{eqnarray}
\end{subequations}

\end{proposition}
\begin{proof}
     We first prove (\ref{CD}). We have\begin{eqnarray}
       &&(x+y)\sum_{k=0}^{m-1}p_k(x)q_k(y)\nonumber\\
        &=&\sum_{k=0}^{m-1}\sum_{j=0}^{k+1}c_{j,k}p_j(x)q_k(y)+\sum_{j=0}^{m-1}\sum_{k=0}^{j+1}d_{j,k}p_j(x)q_k(y)\nonumber\\
        &=&\sum_{k=0}^{m-1}\sum_{j=0}^{m-1}( c_{j,k}+d_{k,j})p_j(x)q_k(y)\nonumber\\
        &&+c_{m,m-1}p_m(x)q_{m-1}(y)+d_{m,m-1}p_{m-1}(x)q_{m}(y)\nonumber\\
        &=&2\sum_{k=0}^{m-1}q_k(y)\sum_{j=0}^{m-1}(\alpha+j+1)p_j(x) \nonumber\\
        &&+\frac{S_{m-1}}{S_m}\left(p_m(x)q_{m-1}(y)+p_{m-1}(x)q_{m}(y)\right) \nonumber\\
        &=&\frac{1}{2 (\alpha +m+1)}\left(\frac{S_{m-1}}{S_m}p_{m-1}(x)-\frac{S_m }{S_{m+1}}p_{m+1}(x)+x p_m(x)\right.\nonumber\\
            &&\left.-\left(\alpha +m+1-\frac{S_m}{S_{m+1}}+\frac{S_{m-1}}{S_m}\right)p_m(x)\right)\left(\frac{S_{m-1}}{S_m}q_{m-1}(y)\right.\nonumber\\
         &&\left.-\frac{S_m }{S_{m+1}}q_{m+1}(y)+y q_m(y)- \left(\alpha+m+1 +\frac{S_m}{S_{m+1}}-\frac{S_{m-1}}{S_m}\right)q_m(y)\right)\nonumber\\
        &&+\frac{S_{m-1}}{S_m}\left(p_m(x)q_{m-1}(y)+p_{m-1}(x)q_{m}(y)\right),
\end{eqnarray}
where the first equality is from (\ref{eq:cjk}) and (\ref{eq:dkj}), the second equality is from (\ref{eq:dterminate}), the third equality is from (\ref{eq:cjkdkj}), (\ref{eq:cmm-1}) and (\ref{eq:dmm-1}), and the last equality is from (\ref{eq:phat}) and (\ref{eq:qcheck}).

We can then derive (\ref{CDpQ}), (\ref{CDPq}), (\ref{CDPQ}) from (\ref{CD}). We give proof of (\ref{CDPq}), and the rest are similar.
Subtract the left-hand side of (\ref{eq:qcheck}) multiplied by \begin{eqnarray}
    \frac{1}{2(\alpha+m+1)}xp_m(x)
\end{eqnarray} from the left-hand side of (\ref{CD}), we have 
\begin{eqnarray}\label{eq:CD-xpmx}
     &&\sum_{k=0}^{m-1}x p_k(x)q_k(y)-xp_m(x) \sum_{k=0}^{m-1}q_k(y) \nonumber\\
     &=&\frac{1}{2 (\alpha +m+1)}\left(\frac{S_{m-1}}{S_m}p_{m-1}(x)-\frac{S_m }{S_{m+1}}p_{m+1}(x)\right.\nonumber\\
            &&\left.-\left(\alpha +m+1-\frac{S_m}{S_{m+1}}+\frac{S_{m-1}}{S_m}\right)p_m(x)\right)\left(\frac{S_{m-1}}{S_m}q_{m-1}(y)\right.\nonumber\\
         &&\left.-\frac{S_m }{S_{m+1}}q_{m+1}(y)+y q_m(y)- \left(\alpha+m+1 +\frac{S_m}{S_{m+1}}-\frac{S_{m-1}}{S_m}\right)q_m(y)\right)\nonumber\\
        &&+\frac{S_{m-1}}{S_m}\left(p_m(x)q_{m-1}(y)+p_{m-1}(x)q_{m}(y)\right)-y\sum_{k=0}^{m-1}p_k(x)q_k(y) .
\end{eqnarray}
Hence we have, by differencing the $z$-copy and $x$-copy of (\ref{eq:CD-xpmx}),
\begin{eqnarray}\label{eq:zx}
     &&\!\!\!\!\!\!\!\!\sum_{k=0}^{m-1}z p_k(z)q_k(y)-zp_m(z) \!\sum_{k=0}^{m-1}q_k(y) -\sum_{k=0}^{m-1}x p_k(x)q_k(y)+xp_m(x) \!\sum_{k=0}^{m-1}q_k(y)\nonumber\\
   &\!\!\!\!\!\!\!\!\!\!\!\!\!\!\!\!=&\!\!\!\!\!\!\!\!\frac{1}{2 (\alpha +m+1)}\left(\frac{S_{m-1}}{S_m}\left(p_{m-1}(z)-p_{m-1}(x)\right)-\frac{S_m }{S_{m+1}}\left(p_{m+1}(z)-p_{m+1}(x)\right)\right.\nonumber\\
            &&\!\!\!\!\!\!\!\!\left.-\left(\alpha +m+1-\frac{S_m}{S_{m+1}}+\frac{S_{m-1}}{S_m}\right)\left(p_m(z)-p_m(x)\right)\right)\left(\frac{S_{m-1}}{S_m}q_{m-1}(y)\right.\nonumber\\
         &&\!\!\!\!\!\!\!\!\left.-\frac{S_m }{S_{m+1}}q_{m+1}(y)+y q_m(y)- \left(\alpha+m+1 +\frac{S_m}{S_{m+1}}-\frac{S_{m-1}}{S_m}\right)q_m(y)\right)\nonumber\\
        &&\!\!\!\!\!\!\!\!+\frac{S_{m-1}}{S_m}\left(\left(p_m(z)-p_m(x)\right)q_{m-1}(y)+\left(p_{m-1}(z)-p_{m-1}(x)\right)q_{m}(y)\right)\nonumber\\  &&\!\!\!\!\!\!\!\!-y\sum_{k=0}^{m-1}(p_k(z)-p_k(x))q_k(y).
\end{eqnarray}
The left-hand side of (\ref{eq:zx}) can be rewritten into
\begin{eqnarray}
     &&\sum_{k=0}^{m-1}z (p_k(z)-p_k(x))q_k(y)-z(p_m(z)-p_m(x)) \!\sum_{k=0}^{m-1}q_k(y)\nonumber\\
     &&+(z-x)\sum_{k=0}^{m-1}p_k(x)q_k(y) - (z-x)p_m(x) \!\sum_{k=0}^{m-1}q_k(y)\nonumber.
\end{eqnarray}
Multiplying 
\begin{equation}
    \frac{1}{z-x}x^\alpha e^{-x}
\end{equation}
 on both sides of (\ref{eq:zx}) and integrating over $x\in[0,\infty)$, the resulting left-hand side is
\begin{eqnarray}
    &&z\sum_{k=0}^{m-1}\tilde{P}_{k-1}(z)q_k(y)-z\tilde{P}_{m-1}(z) \!\sum_{k=0}^{m-1}q_k(y)\nonumber\\
    &&+\sum_{k=0}^{m-1}\int_0^\infty p_k(x)x^\alpha e^{-x} \dd x q_k(y) -\int_0^\infty p_m(x)x^\alpha e^{-x} \dd x \!\sum_{k=0}^{m-1}q_k(y)\nonumber\\
      &=& z\sum_{k=0}^{m-1}\tilde{P}_{k-1}(z)q_k(y)-z\tilde{P}_{m-1}(z) \!\sum_{k=0}^{m-1}q_k(y),
\end{eqnarray}
where we have used the fact (\ref{eq:pint}) such that the last two integrals cancel with each other, and we recall from (\ref{eq:ptilde}) that
\begin{equation}
        \tilde{P}_{k-1}(z)= \int_0^{\infty}\frac{p_k(z)-p_k(x)}{z-x}x^\alpha e^{-x}\dd x.
    \end{equation}
The resulting right-hand side is
\begin{eqnarray}
    &&\frac{1}{2 (\alpha +m+1)}\left(\frac{S_{m-1}}{S_m}\tilde{P}_{m-2}(z)-\frac{S_m }{S_{m+1}}\tilde{P}_{m}(z)\right.\nonumber\\
            &&\left.-\left(\alpha +m+1-\frac{S_m}{S_{m+1}}+\frac{S_{m-1}}{S_m}\right)\tilde{P}_{m-1}(z)\right)\left(\frac{S_{m-1}}{S_m}q_{m-1}(y)\right.\nonumber\\
         &&\left.-\frac{S_m }{S_{m+1}}q_{m+1}(y)+y q_m(y)- \left(\alpha+m+1 +\frac{S_m}{S_{m+1}}-\frac{S_{m-1}}{S_m}\right)q_m(y)\right)\nonumber\\
        &&+\frac{S_{m-1}}{S_m}\left(\tilde{P}_{m-1}(z)q_{m-1}(y)+\tilde{P}_{m-2}(z)q_{m}(y)\right)-y\sum_{k=0}^{m-1}\tilde{P}_{k-1}(z)q_k(y).
\end{eqnarray}

By replacing all the $\tilde{P}_{k-1}(z)$ terms on both sides with  $f_1(z)p_k(z)-P_k(z)$ using (\ref{eq:pkptilde}), one clearly has the $f_1(z)p_k(z)$ terms on left-hand side equals those on the right-hand side from (\ref{CD}), and the terms left on both sides are 
\begin{eqnarray}
  && z\sum_{k=0}^{m-1}P_{k}(z)q_k(y)-zP_{m}(z) \!\sum_{k=0}^{m-1}q_k(y)\nonumber\\
   &=&\frac{1}{2 (\alpha +m+1)}\left(\frac{S_{m-1}}{S_m}P_{m-1}(z)-\frac{S_m }{S_{m+1}}P_{m+1}(z)\right.\nonumber\\
            &&\left.-\left(\alpha +m+1-\frac{S_m}{S_{m+1}}+\frac{S_{m-1}}{S_m}\right)P_{m}(z)\right)\left(\frac{S_{m-1}}{S_m}q_{m-1}(y)\right.\nonumber\\
         &&\left.-\frac{S_m }{S_{m+1}}q_{m+1}(y)+y q_m(y)- \left(\alpha+m+1 +\frac{S_m}{S_{m+1}}-\frac{S_{m-1}}{S_m}\right)q_m(y)\right)\nonumber\\
        &&+\frac{S_{m-1}}{S_m}\left(P_{m}(z)q_{m-1}(y)+P_{m-1}(z)q_{m}(y)\right)-y\sum_{k=0}^{m-1}P_{k}(z)q_k(y).
\end{eqnarray}
Using (\ref{eq:qcheck}) to rewrite the second term on the left-hand side and recollecting the terms completes the proof of (\ref{CDPq}).

The proof of (\ref{CDpQ}) from (\ref{CD}) is similar to that of (\ref{CDPq}) from (\ref{CD}). Use (\ref{eq:phat}) and apply the same method by differencing the $z$-copy and $y$-copy, multiply~(\ref{eq:zy}) and integrate over $y\in[0,\infty)$ with terms rewritten by~(\ref{eq:qtilde}) and~(\ref{eq:f2}), we complete the proof of~(\ref{CDpQ}). Similar to the derivation of~(\ref{CDPq}) from~(\ref{CD}), we can obtain~(\ref{CDPQ}) from~(\ref{CDpQ}) by subtracting the left-hand side of~(\ref{eq:Qcheck}) on the left-hand side of~(\ref{CDpQ}) and rewriting the difference of the $z$-copy and $x$-copy with $\tilde{P}_k(z)$ and $f_1(z)$, while keeping in mind the identity~(\ref{eq:pkptilde}) and the fact~(\ref{eq:pint}). This completes the proof.
\end{proof}

The above results on Christoffel-Darboux formulas for Cauchy-Laguerre biorthogonal kernels give summation-free representations of the Cauchy-Laguerre biorthogonal kernels in terms of certain terms of $p_k(x)$, $q_k(x)$, $P_k(x)$, $Q_k(x)$, whose properties are studied and listed above. We will utilize such summation-free form of Cauchy-Laguerre biorthogonal kernels rather than the summation forms (\ref{eq:01}), (\ref{eq:10}) in the calculation of spectral moments. Such results allow us to represent the derivative of Cauchy-Laguerre biorthogonal kernels in terms of certain derivatives of $p_k(x)$, $q_k(x)$, $P_k(x)$, $Q_k(x)$, which could be recycled to certain $p_k(x)$, $q_k(x)$, $P_k(x)$, $Q_k(x)$ via (\ref{eq:strucrel}) and simplified to a few terms via (\ref{eq:fourtermrec}).

\subsection{Derivative of kernels and recurrence relation for spectral moments}
\label{sec:derivrecspec}
         \begin{proposition}\label{Prop:DerivK}The derivative of Cauchy-Laguerre biorthogonal kernels $K_{01}(x,x)$, $K_{10}(x,x)$ admits 
         \begin{subequations}\label{eq:deriv}
         \begin{equation}\label{K01deriv}
            \!\!\!\!\!\!\!\!  \frac{\dd}{\dd x}x K_{01}(x,x)=\frac{S_{m-1}}{S_{m}}x^{\alpha } e^{-x}  \left(p_m(x) Q_{m-1}(-x)+p_{m-1}(x) Q_m(-x)\right)
         \end{equation}
             \begin{equation}\label{K10deriv}
              \!\!\!\!\!\!\!\!   \frac{\dd}{\dd x}x K_{10}(x,x)=\frac{S_{m-1}}{S_{m}}x^{\alpha +1} e^{-x}  \left(P_m(-x) q_{m-1}(x)+P_{m-1}(-x) q_m(x)\right).
         \end{equation}
         \end{subequations}
         Moreover, the one point kernel admits a representation involving only $q_k(x)$ and $Q_k(-x)$,
            \begin{eqnarray}\label{K01+K10}
                   &&\frac{\dd}{\dd x}\left(x \left( K_{01}(x,x)+K_{10}(x,x)\right)\right)   \nonumber\\
                   &=& -x^{\alpha }e^{-x} \Bigg(a_{1}(q_{m-1}(x) Q_{m-2}(-x)-q_{m-2}(x) Q_{m-1}(-x))\nonumber\\
                   &&+a_{2}( q_{m-1}(x) Q_m(-x)-q_m(x) Q_{m-1}(-x))\nonumber\\&&+a_3(q_{m+1}(x) Q_m(-x)- q_m(x) Q_{m+1}(-x))\nonumber\\
                   &&-x \left(a_4 q_{m-1}(x) Q_{m-1}(-x)+a_5q_{m}(x) Q_{m}(-x)\right)\Bigg),
            \end{eqnarray}
            where 
            \begin{eqnarray*}\label{a1a5}
           a_{1}&=& \frac{(m-1) m (2 \alpha +m) (2 \alpha +m+1) }{4 (\alpha +m) (2 \alpha +2 m-1) (2 \alpha +2 m+1)}\\
               a_2&=& \frac{m (2 \alpha +m+1) \left(4 \alpha ^2+4 \alpha +3 m^2+6 \alpha  m+3 m\right)}{4 (\alpha +m) (\alpha +m+1) (2 \alpha +2 m+1)}\\
            a_3&=&   \frac{m (m+1) (2 \alpha +m+1) (2 \alpha +m+2) }{4 (\alpha +m+1) (2 \alpha +2 m+1) (2 \alpha +2 m+3)}\\
            a_4&=&\frac{m (2 \alpha +m+1) }{(\alpha +m) (2 \alpha +2 m+1)}\\
            a_5&=&\frac{m (2 \alpha +m+1) }{(\alpha +m+1) (2 \alpha +2 m+1)}.
            \end{eqnarray*}
         \end{proposition}
         \begin{proof}
Denote 
\begin{eqnarray}
   P_k'(x)= \frac{\dd}{\dd x}P_k(x) .
\end{eqnarray} Let $y\rightarrow-x$ in (\ref{CDPq}),
     \begin{eqnarray}
      \fl &&\sum _{k=0}^{m-1} P_k(-x) q_k(x)\nonumber \\
      \fl &=&\frac{1}{2 (\alpha +m)}\left(-\left(\frac{4 \alpha ^2+2 \alpha +3 m^2+6 \alpha  m+m}{4 \alpha +4 m+2}+\frac{(m-1) (2 \alpha +m)}{4 \alpha +4 m-2}\right.\right.\fl \nonumber \\ && \left.\left.+x\right)q_{m-1}(x) +\frac{(m-1) (2 \alpha +m) q_{m-2}(x)}{-4 \alpha -4 m+2}+\frac{m (2 \alpha +m+1) q_m(x)}{4 \alpha +4 m+2}\right)\fl \nonumber\\&& \times \left(\left(x-\left(\alpha+m+\frac{1}{2} -\frac{m^2}{4 \alpha +4 m+2}+\frac{(m-1)^2}{4 \alpha +4 m-2}\right)\right) P_{m-1}'(-x)\right.\fl \nonumber \\ &&\left.+\frac{(m-1) (2 \alpha +m) P_{m-2}'(-x)}{-4 \alpha -4 m+2}+\frac{m (2 \alpha +m+1) P_m'(-x)}{4 \alpha +4 m+2}-P_{m-1}(-x)\right)\fl \nonumber\\ &&+\frac{m (2 \alpha +m+1) }{4 \alpha +4 m+2}\left(q_m(x) P_{m-1}'(-x)\right.\left.+q_{m-1}(x) P_m'(-x)\right),
       \end{eqnarray}
and use (\ref{eq:strucP}), (\ref{Qrecrel}), (\ref{qrecrel}), we deduce that
         \begin{subequations}
\begin{eqnarray}\label{cflPq}
    && \sum _{k=0}^{m-1} P_k(-x) q_k(x)\nonumber\\
     &=& \frac{m (2 \alpha +m+1)}{4 x} \Bigg(\frac{\left(2 \alpha  (\alpha +1)+m(2 \alpha + m+1)\right) P_{m-1}(-x) q_m(x)}{(\alpha +m) (2 \alpha +2 m+1)}\nonumber\\&&-\frac{\left(2 \alpha  (\alpha +1)+m^2+2 \alpha  m+m\right) P_m(-x) q_{m-1}(x)}{(\alpha +m+1) (2 \alpha +2 m+1)}\nonumber\\&&+\frac{(m-1) (\alpha +m+1) (2 \alpha +m) P_{m-1}(-x) q_{m-2}(x)}{(\alpha +m) (2 \alpha +2 m-1) (2 \alpha +2 m+1)}\nonumber\\&&+\frac{(1-m) (2 \alpha +m) P_{m-2}(-x) q_{m-1}(x)}{4 (\alpha +m)^2-1}\nonumber\\&&+\frac{2 (m-1) (2 \alpha +m) P_{m-1}(-x) q_{m-1}(x)}{(\alpha +m) (2 \alpha +2 m-1) (2 \alpha +2 m+1)}\nonumber\\&&+\frac{2 (m+1) (2 \alpha +m+2) P_m(-x) q_m(x)}{(\alpha +m+1) (2 \alpha +2 m+1) (2 \alpha +2 m+3)}\nonumber\\&&+\frac{(m+1) (2 \alpha +m+2) P_{m+1}(-x) q_m(x)}{(2 \alpha +2 m+1) (2 \alpha +2 m+3)}\nonumber\\&&-\frac{(m+1) (\alpha +m) (2 \alpha +m+2) P_m(-x) q_{m+1}(x)}{(\alpha +m+1) (2 \alpha +2 m+1) (2 \alpha +2 m+3)}\Bigg)\nonumber\\&&+\frac{m (2 \alpha +m+1) P_{m-1}(-x) q_{m-1}(x)}{2 (\alpha +m)}+\frac{m (2 \alpha +m+1) P_m(-x) q_m(x)}{2 (\alpha +m+1)}\nonumber\\&&-\frac{m (2 \alpha +m+1) }{4 \alpha +4 m+2}\left(P_m(-x) q_{m-1}(x)+P_{m-1}(-x) q_m(x)\right).
\end{eqnarray}
Similarly we have
    \begin{eqnarray}\label{cflpQ}
     && \sum _{k=0}^{m-1} p_k(x) Q_k(-x)\nonumber\\
   &=&-\frac{m (2 \alpha +m+1)}{4 x} \Bigg(\frac{\left(2 \alpha  (\alpha +1)+m(2 \alpha  +m+1)\right) p_{m-1}(x) Q_m(-x)}{(\alpha +m) (2 \alpha +2 m+1)}\nonumber\\&&-\frac{\left(2 \alpha  (\alpha +1)+m^2+2 \alpha  m+m\right) p_m(x) Q_{m-1}(-x)}{(\alpha +m+1) (2 \alpha +2 m+1)}\nonumber\\&&+\frac{(m-1) (\alpha +m+1) (2 \alpha +m) p_{m-1}(x) Q_{m-2}(-x)}{(\alpha +m) (2 \alpha +2 m-1) (2 \alpha +2 m+1)}\nonumber\\&&+\frac{(1-m) (2 \alpha +m) p_{m-2}(x) Q_{m-1}(-x)}{4 (\alpha +m)^2-1}\nonumber\\&&+\frac{2 (m-1) (2 \alpha +m) p_{m-1}(x) Q_{m-1}(-x)}{(\alpha +m) (2 \alpha +2 m-1) (2 \alpha +2 m+1)}\nonumber\\&&+\frac{2 (m+1) (2 \alpha +m+2) p_m(x) Q_m(-x)}{(\alpha +m+1) (2 \alpha +2 m+1) (2 \alpha +2 m+3)}\nonumber\\&&+\frac{(m+1) (2 \alpha +m+2) p_{m+1}(x) Q_m(-x)}{(2 \alpha +2 m+1) (2 \alpha +2 m+3)}\nonumber\\&&-\frac{(m+1) (\alpha +m) (2 \alpha +m+2) p_m(x) Q_{m+1}(-x)}{(\alpha +m+1) (2 \alpha +2 m+1) (2 \alpha +2 m+3)}\Bigg)\nonumber\\&&+\frac{m (2 \alpha +m+1) p_{m-1}(x) Q_{m-1}(-x)}{2 (\alpha +m)}+\frac{m (2 \alpha +m+1) p_m(x) Q_m(-x)}{2 (\alpha +m+1)}\nonumber\\&&-\frac{m (2 \alpha +m+1) }{4 \alpha +4 m+2}\left(p_m(x) Q_{m-1}(-x)+p_{m-1}(x) Q_m(-x)\right).
\end{eqnarray}
\end{subequations}
             Take derivatives in (\ref{cflpQ}) and (\ref{cflPq}), respectively, and use (\ref{eq:strucp}), (\ref{eq:strucQ}), (\ref{precrel}), (\ref{Qrecrel}) and (\ref{eq:strucq}), (\ref{eq:strucP}), (\ref{qrecrel}), (\ref{Precrel}), respectively, we complete the proof of (\ref{K01deriv}) and (\ref{K10deriv}), respectively.
        
     For the derivation of (\ref{K01+K10}), Notice that
            \begin{equation}
         2 \sum _{k=0}^{m-1} q_k(x)=q_{m-1}(x)+p_{m-1}(x),\label{pqcheck}
    \end{equation}
    and insert (\ref{eq:qcheck}) into~(\ref{pqcheck}), we can write $p_k(x)$ in terms of $q_k(x)$,
             \begin{subequations}
               \begin{eqnarray}\label{pintermsq}
         p_{m-1}(x)&=&\frac{1}{\alpha+m+1}x q_{m}(x)+\left(\frac{S_{m-1}}{S_m(\alpha+m+1)}-1\right)q_{m-1}(x)\nonumber\\
         &&+\left(\frac{S_{m-1}}{S_m(\alpha+m+1)}-\frac{S_{m}}{S_{m+1}(\alpha+m+1)} -1\right)q_{m}(x)\nonumber\\
         &&-\frac{S_{m}}{S_{m+1}(\alpha+m+1)} q_{m+1}(x).
    \end{eqnarray}
    Similarly, $P_k(x)$ can be written in terms of $Q_k(x)$,
               \begin{eqnarray}\label{PintermsQ}
               xP_{m-1}(x)&=&\frac{1}{\alpha+m+1}x Q_{m}(x)+\left(\frac{S_{m-1}}{S_m(\alpha+m+1)}-1\right)Q_{m-1}(x)\nonumber\\
         &&+\left(\frac{S_{m-1}}{S_m(\alpha+m+1)}-\frac{S_{m}}{S_{m+1}(\alpha+m+1)} -1\right)Q_{m}(x)\nonumber\\
         &&-\frac{S_{m}}{S_{m+1}(\alpha+m+1)} Q_{m+1}(x).
    \end{eqnarray}
        \end{subequations}
        Insert (\ref{pintermsq}) and (\ref{PintermsQ}) into (\ref{K01deriv}) and (\ref{K10deriv}), respectively, and use (\ref{qrecrel}), (\ref{Qrecrel}), we complete the proof of (\ref{K01+K10}).
        \end{proof}

\begin{theorem}\label{Theorem:Rec}
    The  spectral moments of the Bures-Hall ensemble~(\ref{eq:BHu}) satisfy a three-term recurrence relation
    \begin{eqnarray}\label{specrecrel}
g_1(k)\kappa\!\left(R_{k+2}\right)=g_2(k)\kappa\!\left(R_{k}\right)+g_3(k)\kappa\!\left(R_{k-2}\right),
    \end{eqnarray}
     where the coefficients $g_1(k)$, $g_2(k)$, $g_3(k)$ are in \ref{appendix:A}.
\end{theorem}
\begin{proof}
Perform integration by parts and utilize (\ref{K01+K10}), we have  for $k\neq 0$,
\begin{eqnarray}\label{Rkint}
    \kappa(R_k)&=&\frac{1}{2k}\int_0^{\infty}\frac{\dd}{\dd x}\left(x^k \left( K_{01}(x,x)+K_{10}(x,x)\right)\right) \dd x\nonumber \\
        &=&\frac{1}{2k}\int_0^{\infty} x^{k+\alpha }e^{-x} \Bigg(a_{1}(q_{m-1}(x) Q_{m-2}(-x)-q_{m-2}(x) Q_{m-1}(-x))\nonumber\\
                   &&+a_{2}( q_{m-1}(x) Q_m(-x)-q_m(x) Q_{m-1}(-x))\nonumber\\&&+a_3(q_{m+1}(x) Q_m(-x)- q_m(x) Q_{m+1}(-x))\nonumber\\
                   &&-x \left(a_4 q_{m-1}(x) Q_{m-1}(-x)+a_5q_{m}(x) Q_{m}(-x)\right)\Bigg) \dd x,
\end{eqnarray}
where the coefficients $a_1$ to $a_5$ are as above (\ref{a1a5}).
On the other hand, we utilize the results above on $q_k(x)$ and $Q_k(x)$ to find the derivative of certain terms in integrand of (\ref{Rkint}). For example, using (\ref{eq:strucq}), (\ref{eq:strucQ}), (\ref{qrecrel}), (\ref{Qrecrel}), we have 
        \begin{eqnarray}\label{eq:mm}
            &&\frac{\dd}{\dd x}\left(x(-x^{\alpha } e^{-x}q_{m}(x)Q_{m}(-x))\right)\nonumber\\
            &=&-x^{\alpha }e^{-x}\Bigg(-\frac{2 \left(2 \alpha ^2+3 \alpha +m^2+2 (\alpha +1) m+1\right) q_m(x) Q_m(-x)}{(2 \alpha +2 m+1) (2 \alpha +2 m+3)}\nonumber\\
            &&+\frac{m (2 \alpha +m+1) }{4 \alpha +4 m+2}(q_m(x) Q_{m-1}(-x)+q_{m-1}(x) Q_m(-x))\nonumber\\&&-\frac{(m+1) (2 \alpha +m+2) }{4 \alpha +4 m+6}(q_{m+1}(x) Q_m(-x)+q_m(x) Q_{m+1}(-x))\Bigg).
        \end{eqnarray}
        Do integration by parts between $x^{k-1}$ and $x(-x^{\alpha }e^{-x}q_m(x)Q_m(-x))$, we can represent the integral \begin{eqnarray}
            \int_0^{\infty}x^k \left(-x^{\alpha }e^{-x}\right)q_{m}(x) Q_m(-x)\dd x
        \end{eqnarray} in terms of the integral of other integrands on the right-hand side of (\ref{eq:mm}). Similarly, by doing integration by parts between $x^{k-1}$ and each of the following 7 expressions

        \begin{itemize}
            \item$ x (-x^{\alpha }e^{-x}q_{m-1}(x) Q_{m-1}(-x))$
        \item   $ x(-x^{\alpha }e^{-x}(q_{m-1}(x) Q_m(-x)-q_m(x) Q_{m-1}(-x)))$
  \item  $x (-x^{\alpha }e^{-x}(q_{m-1}(x) Q_{m-2}(-x)-q_{m-2}(x) Q_{m-1}(-x))) $
  \item  $ x (-x^{\alpha }e^{-x}(q_{m+1}(x) Q_m(-x)- q_m(x) Q_{m+1}(-x)))$
   \item     $ x (-x^{\alpha }e^{-x}(q_{m-1}(x) Q_m(-x)+q_m(x) Q_{m-1}(-x)))$
 \item   $  x(-x^{\alpha }e^{-x}(q_{m-1}(x) Q_{m-2}(-x)+q_{m-2}(x) Q_{m-1}(-x)))$
 \item   $ x (-x^{\alpha }e^{-x}(q_{m+1}(x) Q_m(-x)+ q_m(x) Q_{m+1}(-x)))$,
        \end{itemize}
we represent the integral of each term listed above by the remaining terms on the right-hand side of their derivatives, respectively, and do the process repeatedly to represent the integral of each term occurred in the right-hand side of (\ref{Rkint}) by one remaining term to deduce that, when the $R_k$ integral exists, 
\begin{eqnarray}\label{rklast}
        \fl&& \kappa\!\left(R_k\right)\nonumber\\
        \fl &=&\frac{1}{2}\int_0^{\infty}x^k \left(-x^{\alpha }e^{-x}\right)\left(\frac{b_1(k)}{x}+b_2(k) x\right)(q_{m+1}(x) Q_m(-x)+q_m(x) Q_{m+1}(-x)) \dd x,
\end{eqnarray}
where the coefficients $b_1(k)$, $b_2(k)$ are in \ref{appendix:A}. Also, we deduce that the integrand $q_{m+1}(x) Q_m(-x)+ q_m(x) Q_{m+1}(-x)$ satisfies a three-term integral identity
    \begin{eqnarray}\label{intid}
        \fl&&\int_0^{\infty}x^k \left(-x^{\alpha }e^{-x}\right)\frac{1}{x}\left(q_{m+1}(x) Q_m(-x)+ q_m(x) Q_{m+1}(-x)\right)\dd x\nonumber\\
        \fl &=&\int_0^{\infty}x^k \left(-x^{\alpha }e^{-x}\right)\left(c_1(k)x+c_2(k)x^3\right)\left(q_{m+1}(x) Q_m(-x)+ q_m(x) Q_{m+1}(-x)\right)\dd x,
    \end{eqnarray}
    where the coefficients $c_1(k)$, $c_2(k)$ are in \ref{appendix:A}.
    Solving the equations for coefficients of (\ref{rklast}) from $\kappa\!\left(R_{k+2}\right)$, $\kappa\!\left(R_k\right)$, $\kappa\!\left(R_{k-2}\right)$ in terms of the integral identity (\ref{intid}) yields (\ref{specrecrel}).
\end{proof}
\section{Application to entanglement entropy}\label{App}

The recurrence relation for spectral moments in terms of a real order $k$ is the key building block to a systematic way of obtaining any order cumulant of entanglement entropy recursively. In the following, we derive a recurrence relation of the mean of linear statistics $T_k$ as in (\ref{eq:Tk}), which is recalled here,
\begin{eqnarray}
    T_k=\sum_{i=1}^m x_i^k \ln x_i
\end{eqnarray}
 over the unconstrained Bures-Hall ensemble~(\ref{eq:BHu}) from (\ref{specrecrel}), and apply the recurrence relations to reproduce the results on the mean of von Neumann entropy and quantum purity over Bures-Hall ensemble~(\ref{eq:BH})~\cite{Sarkar19,Wei20BHA,LW21}.

The von Neumann entropy of a bipartite system is defined as 
\begin{equation}\label{eq:vN}
S=-\tr\left(\rho_{A}\ln\rho_{A}\right)=-\sum_{i=1}^{m}\lambda_{i}\ln\lambda_{i}
\end{equation}
supported in $S\in\left[0,\ln{m}\right]$, which achieves the maximally-entangled state $S=\ln{m}$ when 
\begin{equation}
\lambda_{1}=\dots=\lambda_{m}=\frac{1}{m},
\end{equation} and degenerates to the separable state $S=0$ when 
\begin{equation}
\lambda_{1}=1,~~\lambda_{2}=\dots=\lambda_{m}=0.
\end{equation}
In particular, the von Neumann entropy over the unconstrained Bures-Hall ensemble (\ref{eq:BHu}) is the linear statistics 
\begin{eqnarray}
    T_1=\sum_{i=1}^{m} x_i \ln x_i
\end{eqnarray}over the unconstrained Bures-Hall ensemble (\ref{eq:BHu}). As in~\cite{Wei20BHA,Wei20BH,wei2025skewnessvonneumannentropy}, we have the relation
\begin{eqnarray}\label{momentconvS}
    \kappa(S)= \psi _0\!\left(\alpha  m+\frac{1}{2} m(m+1)+1\right)-\frac{2}{m(2\alpha+m+1)}\kappa(T_1).
\end{eqnarray}
\begin{proposition} We have recurrence relation for $T_k$ as
\begin{eqnarray}\label{Tkrecrel}
     \fl&&   g_1(k) \kappa\!\left(T_{k+2}\right)\nonumber\\
     \fl &=&g_2(k)\kappa\!\left(T_{k}\right)+g_3(k)\kappa\!\left(T_{k-2}\right)-g_1'(k)\kappa\!\left(R_{k+2}\right)+g_2'(k)\kappa\!\left(R_{k}\right)+g_3'(k)\kappa\!\left(R_{k-2}\right),
\end{eqnarray}
   where coefficients $g_1(k)$, $g_2(k)$, $g_3(k)$ are as above and in \ref{appendix:A}, and  
   \begin{eqnarray}
       g_i'(k)=\frac{\dd }{\dd k}g_i(k)
   \end{eqnarray}
   for $i=1,2,3$.
As a consequence, the mean of von Neumann entropy of Bures-Hall ensemble (\ref{eq:BH}) is
\begin{eqnarray}\label{entropy}
         \kappa(S)=\psi_0\!\left(m n-\frac{m^2}{2}+1\right)-\psi_0\!\left(n+\frac{1}{2}\right).
\end{eqnarray}
\end{proposition}
\begin{proof}
The result (\ref{Tkrecrel}) is obtained by taking derivative with respect to $k$ in (\ref{specrecrel}).
We now calculate the initial conditions. Firstly, by definition we have
\begin{eqnarray}\label{R0}
\kappa\!\left(R_0\right)=\kappa\left(\sum_{i=1}^m x_i^0\right)=m.
\end{eqnarray}
Also, we can use (\ref{K01deriv}), (\ref{K10deriv}) to calculate certain $\kappa(R_k)$ for $k\neq 0$ as initial conditions
\begin{eqnarray}
  \kappa(R_k)&=&-\frac{1}{2k} \int_0^{\infty} \frac{S_{m-1}}{S_{m}}x^{k+\alpha } e^{-x} (p_m(x) Q_{m-1}(-x)+p_{m-1}(x) Q_m(-x)\nonumber\\
  &&+x\left(P_m(-x) q_{m-1}(x)+P_{m-1}(-x) q_m(x)\right)) \dd x .
\end{eqnarray}
Insert (\ref{pqpoly}), (\ref{eq:kerM}) for the terms in integrand, where the integral of Meijer-G functions in $P_k(-x)$ and $Q_k(-x)$ can be calculated utilizing Mellin transform as in \cite{Wei20BHA,Wei20BH,wei2025skewnessvonneumannentropy}, we obtain that
\begin{subequations}
\begin{eqnarray}\label{R-1}
\!\!\!\! \!\!\!\!\!\!  \kappa\!\left(R_{-1}\right)&=&  \frac{m (2 \alpha +m+1)}{2 \alpha  (\alpha +1)}
\end{eqnarray}    
\begin{eqnarray}\label{R-2}
 \!\!\!\! \!\!\!\!\!\!  \kappa\!\left(R_{-2}\right)&=& \frac{m (2 \alpha +m+1) \left(4 \alpha ^2+4 \alpha +6 \alpha  m+3 m^2+3 m-2\right)}{4 (\alpha -1) \alpha  (\alpha +1) (\alpha +2) (2 \alpha +1)}
\end{eqnarray}
\begin{eqnarray}\label{R-3}
 \!\!\!\!\!\!\!\!\!\!    \kappa\!\left(R_{-3}\right)&=&\frac{m (\alpha +m-1) (\alpha +m) (\alpha +m+1) (\alpha +m+2) (2 \alpha +m+1)}{2 (\alpha -2) (\alpha -1) \alpha ^2 (\alpha +1)^2 (\alpha +2) (\alpha +3)}.
\end{eqnarray}
\end{subequations}
Now using (\ref{specrecrel}) for $k=-1$ and insert (\ref{R-1}), (\ref{R-3}), we have
\begin{eqnarray}\label{R1}
    \kappa\!\left(R_{1}\right)=\frac{1}{2} m (2 \alpha +m+1).
\end{eqnarray}
The initial condition $\kappa\!\left(T_0\right)$ has been recently obtained in \cite{wei2025averagerelativeentropyrandom} in terms of the constrained Bures-Hall ensemble, and similarly we have here for unconstrained Bures-Hall ensemble as 
\begin{eqnarray}
    \kappa\!\left(T_0\right) &=&- \left(\alpha +\frac{1}{2}\right) \psi_0 (\alpha +1)-(2 \alpha +1) \psi_0 (2 \alpha +1)+\left(\alpha +m+\frac{1}{2}\right)\nonumber\\ &&\times  \psi_0 (\alpha +m+1)+2 (2 \alpha +m+1) \psi_0 (2 \alpha +m+1)\nonumber \\ &&-(2 \alpha +2 m+1) \psi_0 (2 \alpha +2 m+1)-m.
\end{eqnarray}
We can take derivative with respect to $k$ in the calculation of $\kappa\!\left(R_{k}\right)$ and evaluate at certain $k$, which yields the initial conditions
\begin{subequations}
    \begin{eqnarray}\label{T-1}
        \kappa\!\left(T_{-1}\right) &=&\left(1+\frac{ 3 m(2 \alpha +m+1) }{2 \alpha  (\alpha +1)}\right)\psi_0 (\alpha +1)-\frac{(\alpha +m) (\alpha +m+1) }{\alpha  (\alpha +1)}\nonumber\\&&\times\psi_0 (\alpha +m+1)+\frac{\left(\alpha ^2+\alpha -1\right) m (2 \alpha +m+1)}{2 \alpha ^2 (\alpha +1)^2}
    \end{eqnarray}
    
    \begin{eqnarray}\label{T-2}
        \kappa\!\left(T_{-2}\right) &=&\frac{m (2 \alpha +m+1) \left(4 \alpha  (\alpha +1)+3 m(2 \alpha +m+1) -2\right)}{4 (\alpha -1) \alpha  (\alpha +1) (\alpha +2) (2 \alpha +1)}\nonumber\\&&\times (\psi_0 (\alpha +1)+2 \psi_0 (2 \alpha +1)-2 \psi_0 (2 \alpha +m+1))+N_{-2}
    \end{eqnarray}
    
    \begin{eqnarray}\label{T-3}
        \kappa\!\left(T_{-3}\right) &=&\frac{m (\alpha +m-1) (\alpha +m) (\alpha +m+1) (\alpha +m+2) (2 \alpha +m+1) }{2 (\alpha -2) (\alpha -1) \alpha ^2 (\alpha +1)^2 (\alpha +2) (\alpha +3)}\nonumber\\&& \times (3 \psi_0 (\alpha +1)-2 \psi_0 (\alpha +m+1))+N_{-3},
    \end{eqnarray}
    where \begin{eqnarray*}
      \!\!\!\!\!\!\!\!\!  N_{-2}&=&-\frac{1}{8 (\alpha -1)^2 \alpha ^2 (\alpha +1)^2 (\alpha +2)^2 (2 \alpha +1)^2}\nonumber\\&& \times(18 \alpha ^5 m^4+39 \alpha ^4 m^4+36 \alpha ^3 m^4+33 \alpha ^2 m^4-6 \alpha  m^4-12 m^4\nonumber\\&& +48 \alpha ^6 m^3+132 \alpha ^5 m^3+222 \alpha ^4 m^3+264 \alpha ^3 m^3+66 \alpha ^2 m^3-60 \alpha  m^3\nonumber\\&& -24 m^3+32 \alpha ^7 m^2+104 \alpha ^6 m^2+290 \alpha ^5 m^2+501 \alpha ^4 m^2+320 \alpha ^3 m^2\nonumber\\&& -5 \alpha ^2 m^2-50 \alpha  m^2-4 m^2+88 \alpha ^6 m+264 \alpha ^5 m+238 \alpha ^4 m\nonumber\\&& +36 \alpha ^3 m-6 \alpha ^2 m+20 \alpha  m+8 m)
    \end{eqnarray*}
    
    \begin{eqnarray*}
\!\!\!\!\!\!\!\!\! N_{-3}&=&-\frac{m (2 \alpha +m+1)}{24 (\alpha -2)^2 (\alpha -1)^2 \alpha ^3 (\alpha +1)^3 (\alpha +2)^2 (\alpha +3)^2}\nonumber\\&& \times(88 \alpha ^8+352 \alpha ^7+40 \alpha ^6-1112 \alpha ^5-608 \alpha ^4+1048 \alpha ^3+480 \alpha ^2\nonumber\\&& -288 \alpha +5 \alpha ^6 m^4+15 \alpha ^5 m^4+89 \alpha ^4 m^4+153 \alpha ^3 m^4-346 \alpha ^2 m^4\nonumber\\&& -420 \alpha  m^4+216 m^4+20 \alpha ^7 m^3+70 \alpha ^6 m^3+386 \alpha ^5 m^3+790 \alpha ^4 m^3\nonumber\\&& -1078 \alpha ^3 m^3-2372 \alpha ^2 m^3+24 \alpha  m^3+432 m^3+26 \alpha ^8 m^2+104 \alpha ^7 m^2\nonumber\\&& +603 \alpha ^6 m^2+1445 \alpha ^5 m^2-1083 \alpha ^4 m^2-4453 \alpha ^3 m^2-1022 \alpha ^2 m^2\nonumber\\&& +1428 \alpha  m^2-216 m^2+12 \alpha ^9 m+54 \alpha ^8 m+388 \alpha ^7 m+1106 \alpha ^6 m\nonumber\\&& -338 \alpha ^5 m-3484 \alpha ^4 m-1262 \alpha ^3 m+2108 \alpha ^2 m+120 \alpha  m-432 m).
    \end{eqnarray*}
\end{subequations}
    
    Let $k=-1$ in (\ref{Tkrecrel}), we have
    \begin{eqnarray}\label{Tk-1}
\fl~~~~~~  g_1(-1)\kappa\!\left(T_{1}\right)&=&g_2(-1)\kappa\!\left(T_{-1}\right)+g_3(-1)\kappa\!\left(T_{-3}\right)-g_1'(-1)\kappa\!\left(R_{1}\right)+g_2'(-1)\kappa\!\left(R_{-1}\right)\nonumber\\
\fl~~~~~  &&+g_3'(-1)\kappa\!\left(R_{-3}\right),
    \end{eqnarray}
    where
    \begin{eqnarray*}
        \fl g_1(-1)&=&96  (\alpha +m-1) (\alpha +m) (\alpha +m+1) (\alpha +m+2)\\
      \fl g_2(-1)&=&   -144 m (\alpha +m-1) (\alpha +m) (\alpha +m+1) (\alpha +m+2) (2 \alpha +m+1)\\ 
       \fl g_3(-1)&=& 48 (\alpha -2) (\alpha -1) \alpha  (\alpha +1) (\alpha +2) (\alpha +3) \left(2 \alpha  (\alpha +1)+3 m(2 \alpha +m+1) \right)\\
       \fl  g_1'(-1)&=&8 (2 (\alpha  (\alpha +1) (5 \alpha  (\alpha +1)+3)-6)+9 m^4+18 (2 \alpha +1) m^3\\ \fl &&+5 (11 \alpha  (\alpha +1)+3) m^2+(2 \alpha +1) (19 \alpha  (\alpha +1)+6) m)\\
       \fl g_2'(-1)&=& 4 (-8 \alpha ^2 (\alpha +1)^2 \left(\alpha ^2+\alpha -2\right)+27 m^6+81 (2 \alpha +1) m^5\nonumber\\ \fl &&+(366 \alpha  (\alpha +1)-45) m^4+3 (2 \alpha +1) (64 \alpha  (\alpha +1)-75) m^3\nonumber\\ \fl &&+6 (\alpha  (\alpha +1) (29 \alpha  (\alpha +1)-94)-15) m^2\nonumber\\ \fl &&+6 (2 \alpha +1) \left(\alpha  (\alpha +1) \left(\alpha ^2+\alpha -17\right)+6\right) m)\\
       \fl g_3'(-1)&=&8 (2 \alpha  (\alpha +1) (\alpha  (\alpha +1) (\alpha  (\alpha +1) (7 \alpha  (\alpha +1)+1)-156)+108)\nonumber\\ \fl &&+3 (\alpha  (\alpha +1) (\alpha  (\alpha +1) (4 \alpha  (\alpha +1)+19)-144)+36) m^2\nonumber\\ \fl &&+3 (2 \alpha +1) (\alpha  (\alpha +1) (\alpha  (\alpha +1) (4 \alpha  (\alpha +1)+19)-144)+36) m).
    \end{eqnarray*}
     Insert (\ref{T-1}), (\ref{T-3}), (\ref{R1}), (\ref{R-1}), (\ref{R-3}) into~(\ref{Tk-1}), we deduce that  \begin{eqnarray}\label{T1}
        \kappa\!\left(T_1\right) &=& \frac{1}{2} m (2 \alpha +m+1) \psi_0 (\alpha +m+1).
    \end{eqnarray}
    Perform the conversion (\ref{momentconvS}) as in \cite{Wei20BHA,Wei20BH,wei2025skewnessvonneumannentropy}, we obtain
    \begin{eqnarray}\label{Salpha}
       \kappa(S)= \psi _0\!\left(\alpha  m+\frac{1}{2} m(m+1)+1\right)-\psi _0(\alpha +m+1).
    \end{eqnarray}
      Evaluate $\alpha$ in (\ref{Salpha}) as (\ref{eq:aBH}), we have (\ref{entropy}).
\end{proof}

This reproduces the result on the mean of von Neumann entropy conjectured in \cite{Sarkar19} and first proved via summation method in \cite{Wei20BHA}.

Another application of the recurrence relation for spectral moments is the 
quantum purity, which was previously obtained in \cite{Sarkar19,Wei20BHA,LW21}. The quantum purity is defined as 
\begin{equation}\label{eq:P}
\mathrm{P}=\sum_{i=1}^{m}\lambda_{i}^{2},
\end{equation}
which is supported in $\mathrm{P} \in [1/m,1]$, and achieves maximally-entangled state and separable state when $\mathrm{P}=1/m$ and $\mathrm{P}=1$, respectively. In particular, the  purity is the linear statistics 
\begin{eqnarray}
    R_2=\sum_{i=1}^{m} x_i^2,
\end{eqnarray}over the unconstrained Bures-Hall ensemble (\ref{eq:BHu}). As in~\cite{Wei20BHA,LW21}, we have the relation
\begin{eqnarray}\label{momentconvP}
    \kappa(\mathrm{P})=\frac{\Gamma(d)}{\Gamma(d+2)}\kappa(R_2),
\end{eqnarray}where $d$ is in (\ref{eq:d}).

\begin{proposition} We have
\begin{eqnarray}\label{R0rec}
 g_1(0) \kappa\!\left(R_{2}\right) =g_2(0)\kappa\!\left(R_{0}\right)+ g_3(0)\kappa\!\left(R_{-2}\right),
\end{eqnarray}
where \begin{eqnarray*}
    \fl g_1(0) &=&6 (2 \alpha +2 m+1)^2 \left(4 \alpha  (\alpha +1)+3m(2 \alpha +m+1) -2\right)\\
      \fl g_2(0)&=& 3 m (2 \alpha +m+1) \left(4 \alpha  (\alpha +1)+5m(2 \alpha +m+1) +2\right) \nonumber\\ \fl && \times \left(4 \alpha  (\alpha +1)+3 m(2 \alpha +m+1) -2\right)\\
    \fl g_3(0)&=& 6 (\alpha -1) \alpha  (\alpha +1) (\alpha +2) (2 \alpha +1)^2 \left(4 \alpha  (\alpha +1)+5 m (2 \alpha+m +1) +2\right).
\end{eqnarray*}
 As a consequence, the mean of quantum purity of Bures-Hall ensemble (\ref{eq:BH}) is
 \begin{eqnarray}\label{purity}
 \kappa(\mathrm{P})=\frac{m^2-2 m n-4 n^2-1}{2 n\left(m^2-2 m n-2\right)} .
 \end{eqnarray}
\end{proposition}
\begin{proof}
Let $k=0$ in (\ref{specrecrel}), we have (\ref{R0rec}).
Insert (\ref{R0}), (\ref{R-2}) into (\ref{R0rec}), we deduce that
     \begin{eqnarray}\label{R2}
        \kappa\!\left(R_2\right)=\frac{m (2 \alpha +m+1) \left(4 \alpha ^2+4 \alpha +10 \alpha  m+5 m^2+5 m+2\right)}{4 (2 \alpha +2 m+1)}.
    \end{eqnarray}
    Perform the conversion (\ref{momentconvP}) as in \cite{Wei20BHA,LW21}, we obtain that
    \begin{eqnarray}\label{Palpha}
       \kappa(\mathrm{P})=  \frac{4 \alpha ^2+4 \alpha+10 \alpha  m +5 m^2+5 m+2}{(2 \alpha +2 m+1) \left(m^2+2 \alpha  m+m+2\right)} .
    \end{eqnarray}
    Evaluate $\alpha$ in (\ref{Palpha}) as (\ref{eq:aBH}), we obtain 
    (\ref{purity}).
\end{proof}

This reproduces the mean formula~\cite{Sarkar19,Wei20BHA,LW21} of quantum purity for Bures-Hall ensemble.

\section*{Acknowledgment} 
The work of Lu Wei was supported by the U.S. National Science Foundation (2306968) and the U.S. Department of Energy (DE-SC0024631).

\appendix
\section{Coefficients of certain expressions}\label{appendix:A}
  \begin{eqnarray}
             \fl && g_1(k)\nonumber\\ \fl &=& (k+2) (k+3) (-2 \alpha +k-2 m-1) (2 \alpha +k+2 m+1)\nonumber\\ \fl &&  \times( -2 k^3+10 k^2+k(8 \alpha ^2+8 \alpha +9 m^2+18 \alpha  m+9 m-10)\nonumber\\ \fl &&-4 \alpha ^2-4 \alpha -3 m^2-6 \alpha  m-3 m+2)
        \end{eqnarray}
            \begin{eqnarray}
                \fl && g_2(k)\nonumber\\ \fl &=&-4 k^9+4 k^8+2 k^7 (14 \alpha ^2+14 \alpha +9 m^2+18 \alpha  m+9 m+18)+2 k^6 (22 \alpha ^2+22 \alpha\nonumber\\ \fl && +39 m^2+78 \alpha  m+39 m+6)+2 k^5 (-16 \alpha ^4-32 \alpha ^3-60 \alpha ^2-44 \alpha +27 m^4+108 \alpha  m^3\nonumber\\ \fl &&+54 m^3+117 \alpha ^2 m^2+117 \alpha  m^2+21 m^2+18 \alpha ^3 m+27 \alpha ^2 m-3 \alpha  m-6 m-30)\nonumber\\ \fl &&-2 k^4 (108 \alpha ^4+216 \alpha ^3+176 \alpha ^2+68 \alpha +135 m^4+540 \alpha  m^3+270 m^3+789 \alpha ^2 m^2\nonumber\\ \fl &&+789 \alpha  m^2+249 m^2+498 \alpha ^3 m+747 \alpha ^2 m+477 \alpha  m+114 m+18)\nonumber\\ \fl &&-k^3 (64 \alpha ^6+192 \alpha ^5+408 \alpha ^4+496 \alpha ^3+172 \alpha ^2-44 \alpha +243 m^6+1458 \alpha  m^5+729 m^5\nonumber\\ \fl &&+3456 \alpha ^2 m^4+3456 \alpha  m^4+1275 m^4+4104 \alpha ^3 m^3+6156 \alpha ^2 m^3+4722 \alpha  m^3+1335 m^3\nonumber\\ \fl &&+2520 \alpha ^4 m^2+5040 \alpha ^3 m^2+6006 \alpha ^2 m^2+3486 \alpha  m^2+564 m^2+720 \alpha ^5 m+1800 \alpha ^4 m\nonumber\\ \fl &&+2964 \alpha ^3 m+2646 \alpha ^2 m+798 \alpha  m+18 m-28)-k^2 (96 \alpha ^6+288 \alpha ^5+264 \alpha ^4+48 \alpha ^3\nonumber\\ \fl &&-116 \alpha ^2-92 \alpha +405 m^6+2430 \alpha  m^5+1215 m^5+5724 \alpha ^2 m^4+5724 \alpha  m^4+1077 m^4\nonumber\\ \fl &&+6696 \alpha ^3 m^3+10044 \alpha ^2 m^3+3606 \alpha  m^3+129 m^3+4008 \alpha ^4 m^2+8016 \alpha ^3 m^2\nonumber\\ \fl &&+4158 \alpha ^2 m^2+150 \alpha  m^2-300 m^2+1104 \alpha ^5 m+2760 \alpha ^4 m+1884 \alpha ^3 m+66 \alpha ^2 m\nonumber\\ \fl &&-486 \alpha  m-162 m-20)+k (-32 \alpha ^6-96 \alpha ^5-40 \alpha ^4+80 \alpha ^3+72 \alpha ^2+16 \alpha +27 m^6\nonumber\\ \fl &&+162 \alpha  m^5+81 m^5+384 \alpha ^2 m^4+384 \alpha  m^4+141 m^4+456 \alpha ^3 m^3+684 \alpha ^2 m^3+522 \alpha  m^3\nonumber\\ \fl &&+147 m^3+240 \alpha ^4 m^2+480 \alpha ^3 m^2+660 \alpha ^2 m^2+420 \alpha  m^2+72 m^2+240 \alpha ^3 m+360 \alpha ^2 m\nonumber\\ \fl &&+144 \alpha  m+12 m)+3 m (32 \alpha ^5+80 \alpha ^4+64 \alpha ^3+16 \alpha ^2-8 \alpha +15 m^5+90 \alpha  m^4+45 m^4\nonumber\\ \fl &&+212 \alpha ^2 m^3+212 \alpha  m^3+41 m^3+248 \alpha ^3 m^2+372 \alpha ^2 m^2+138 \alpha  m^2+7 m^2+144 \alpha ^4 m\nonumber\\ \fl &&+288 \alpha ^3 m+160 \alpha ^2 m+16 \alpha  m-8 m-4)
            \end{eqnarray}
          \begin{eqnarray}
         \fl && g_3(k)\nonumber\\ \fl &=&     (k-2) (k-1) (-2 \alpha +k-1) (-\alpha +k-2) (-\alpha +k-1) (\alpha +k-1) (\alpha +k)  \nonumber\\ \fl &&  \times(2 \alpha +k+1) ( 2 k^3+2 k^2-k (8 \alpha ^2+8 \alpha +9 m^2+18 \alpha  m+9 m+6) \nonumber\\ \fl &&-3 (4 \alpha ^2+4 \alpha +5 m^2+10 \alpha  m+5 m+2))
        \end{eqnarray}
        
\begin{eqnarray}
     \fl && b_1(k)\nonumber\\ \fl &=& -\frac{1}{D_{b_1(k)}}(k-1) (-2 \alpha +k-1) (-\alpha +k-1) (\alpha +k) (2 \alpha +k+1)\nonumber \\ \fl &&\times(2 k^3 \left(2 \alpha ^2+5 \alpha +2 m^2+4 \alpha  m+5 m+3\right)\nonumber \\ \fl &&+k^2 \left(4 \alpha ^3+22 \alpha ^2+38 \alpha +4 m^3+12 \alpha  m^2+23 m^2+12 \alpha ^2 m+46 \alpha  m+39 m+20\right)\nonumber \\ \fl &&+k (8 \alpha ^3+30 \alpha ^2+40 \alpha +2 m^4+8 \alpha  m^3+14 m^3+10 \alpha ^2 m^2+40 \alpha  m^2+37 m^2+4 \alpha ^3 m\nonumber \\ \fl &&+34 \alpha ^2 m+70 \alpha  m+43 m+18)+(m+1)^2 (2 \alpha +m+2)^2),
\end{eqnarray}
where
\begin{eqnarray*}
       \fl && D_{b_1(k)}\nonumber\\ \fl&=&2 (\alpha +m+1) (\alpha +m+2)( 2 k^5 (2 \alpha +2 m+3)^2 +4 k^4 (m+1) (2 \alpha +m+2)\nonumber \\ \fl && -k^3 (32 \alpha ^4+200 \alpha ^3+472 \alpha ^2+478 \alpha +36 m^4+144 \alpha  m^3+216 m^3+212 \alpha ^2 m^2\nonumber \\ \fl && +644 \alpha  m^2+483 m^2+136 \alpha ^3 m+628 \alpha ^2 m+954 \alpha  m+477 m+178)\nonumber \\ \fl && -k^2 (m+1) (32 \alpha ^3+136 \alpha ^2+186 \alpha +18 m^3+72 \alpha  m^2+90 m^2+88 \alpha ^2 m+232 \alpha  m\nonumber \\ \fl && +149 m+82)+k(32 \alpha ^4+128 \alpha ^3+192 \alpha ^2+126 \alpha +4 m^4+16 \alpha  m^3+24 m^3+48 \alpha ^2 m^2\nonumber \\ \fl && +96 \alpha  m^2+59 m^2+64 \alpha ^3 m+192 \alpha ^2 m+190 \alpha  m+69 m+32)+(m+1)\nonumber \\ \fl && \times(8 \alpha ^3+24 \alpha ^2+26 \alpha +2 m^3+8 \alpha  m^2+10 m^2+12 \alpha ^2 m+28 \alpha  m+17 m+10)
\end{eqnarray*}

 \begin{eqnarray}
         \fl &&  b_2(k)\nonumber\\ \fl&=&-\frac{1}{D_{b_2(k)}}(k+2) (2 \alpha +k+2 m+3)(2 k^4(2 \alpha ^2+5 \alpha +2 m^2+4 \alpha  m+5 m+3)\nonumber \\ \fl && -k^3 (12 \alpha ^3+50 \alpha ^2+66 \alpha +12 m^3+36 \alpha  m^2+49 m^2+36 \alpha ^2 m+98 \alpha  m+65 m+28)\nonumber \\ \fl && +k^2(8 \alpha ^4+68 \alpha ^3+162 \alpha ^2+150 \alpha +18 m^4+72 \alpha  m^3+98 m^3+98 \alpha ^2 m^2+284 \alpha  m^2\nonumber \\ \fl && +191 m^2+52 \alpha ^3 m+254 \alpha ^2 m+362 \alpha  m+159 m+48)-k (24 \alpha ^4+116 \alpha ^3+198 \alpha ^2\nonumber \\ \fl && +142 \alpha +12 m^5+60 \alpha  m^4+84 m^4+108 \alpha ^2 m^3+324 \alpha  m^3+224 m^3+84 \alpha ^3 m^2 \nonumber \\ \fl && +424 \alpha ^2 m^2+628 \alpha  m^2+283 m^2+24 \alpha ^4 m+212 \alpha ^3 m+528 \alpha ^2 m+510 \alpha  m\nonumber \\ \fl && +167 m+36)+16 \alpha ^4+60 \alpha ^3+82 \alpha ^2+48 \alpha +4 m^5+20 \alpha  m^4+26 m^4+44 \alpha ^2 m^3\nonumber \\ \fl && +108 \alpha  m^3+66 m^3+52 \alpha ^3 m^2+182 \alpha ^2 m^2+212 \alpha  m^2+81 m^2+24 \alpha ^4 m+128 \alpha ^3 m\nonumber \\ \fl && +230 \alpha ^2 m+174 \alpha  m+47 m+10),
         \end{eqnarray}
         where
         \begin{eqnarray*}
          \fl &&     D_{b_2(k)}\nonumber\\ \fl &=&2 (\alpha +m+1) (\alpha +m+2)(-2 k^5 (2 \alpha +2 m+3)^2-4 k^4 (m+1) (2 \alpha +m+2)\nonumber \\ \fl && +k^3(32 \alpha ^4+200 \alpha ^3+472 \alpha ^2+478 \alpha +36 m^4+144 \alpha  m^3+216 m^3+212 \alpha ^2 m^2\nonumber \\ \fl && +644 \alpha  m^2+483 m^2+136 \alpha ^3 m+628 \alpha ^2 m+954 \alpha  m+477 m+178) +k^2 (m+1) \nonumber \\ \fl &&\times(32 \alpha ^3+136 \alpha ^2+186 \alpha +18 m^3+72 \alpha  m^2+90 m^2+88 \alpha ^2 m+232 \alpha  m+149 m+82)\nonumber \\ \fl && -k (32 \alpha ^4+128 \alpha ^3+192 \alpha ^2+126 \alpha +4 m^4+16 \alpha  m^3+24 m^3+48 \alpha ^2 m^2+96 \alpha  m^2\nonumber \\ \fl &&+59 m^2+64 \alpha ^3 m+192 \alpha ^2 m+190 \alpha  m+69 m+32)-(m+1) (8 \alpha ^3+24 \alpha ^2+26 \alpha \nonumber \\ \fl &&+2 m^3+8 \alpha  m^2+10 m^2+12 \alpha ^2 m+28 \alpha  m+17 m+10))
         \end{eqnarray*}
         
         \begin{eqnarray}
              \fl &&c_1(k)\nonumber\\ \fl &=&\frac{1}{D_{c_1(k)}} (4 k^{11} (2 \alpha +2 m+3)^2+8 k^{10}(24 \alpha ^2+74 \alpha +25 m^2+50 \alpha  m+75 m+56)\nonumber \\ \fl &&-2 k^9(56 \alpha ^4+296 \alpha ^3+134 \alpha ^2-948 \alpha +36 m^4+144 \alpha  m^3+216 m^3+236 \alpha ^2 m^2\nonumber \\ \fl &&+668 \alpha  m^2-25 m^2+184 \alpha ^3 m+748 \alpha ^2 m+10 \alpha  m-1047 m-954)\nonumber \\ \fl &&-2 k^8 (560 \alpha ^4+3016 \alpha ^3+5044 \alpha ^2+1778 \alpha +378 m^4+1512 \alpha  m^3+2268 m^3\nonumber \\ \fl &&+2460 \alpha ^2 m^2+6996 \alpha  m^2+3753 m^2+1896 \alpha ^3 m+7764 \alpha ^2 m+8082 \alpha  m+1053 m\nonumber \\ \fl &&-1050)-2 k^7 (-64 \alpha ^6-248 \alpha ^5+2408 \alpha ^4+14830 \alpha ^3+29456 \alpha ^2+22870 \alpha +108 m^6\nonumber \\ \fl &&+648 \alpha  m^5+972 m^5+1440 \alpha ^2 m^4+4680 \alpha  m^4+5289 m^4+1440 \alpha ^3 m^3+7920 \alpha ^2 m^3\nonumber \\ \fl &&+20076 \alpha  m^3+17154 m^3+548 \alpha ^4 m^2+5416 \alpha ^3 m^2+26755 \alpha ^2 m^2+49715 \alpha  m^2\nonumber \\ \fl &&+28497 m^2-56 \alpha ^5 m+956 \alpha ^4 m+14486 \alpha ^3 m+47611 \alpha ^2 m+56613 \alpha  m+21420 m\nonumber \\ \fl &&+5496)-2 k^6 (-512 \alpha ^6-2048 \alpha ^5+5704 \alpha ^4+45874 \alpha ^3+96580 \alpha ^2+84892 \alpha\nonumber \\ \fl && +918 m^6+5508 \alpha  m^5+8262 m^5+12186 \alpha ^2 m^4+39726 \alpha  m^4+34800 m^4\nonumber \\ \fl &&+12024 \alpha ^3 m^3+66780 \alpha ^2 m^3+129696 \alpha  m^3+84870 m^3+4424 \alpha ^4 m^2+44920 \alpha ^3 m^2\nonumber \\ \fl &&+160015 \alpha ^2 m^2+234179 \alpha  m^2+119955 m^2-512 \alpha ^5 m+7568 \alpha ^4 m+71090 \alpha ^3 m\nonumber \\ \fl &&+195697 \alpha ^2 m+221385 \alpha  m+89487 m+27234)+k^5 (256 \alpha ^8+4416 \alpha ^7+32256 \alpha ^6\nonumber \\ \fl &&+115312 \alpha ^5+198000 \alpha ^4+90880 \alpha ^3-185336 \alpha ^2-267104 \alpha +972 m^8+7776 \alpha  m^7\nonumber \\ \fl &&+11664 m^7+26460 \alpha ^2 m^6+80892 \alpha  m^6+54177 m^6+49896 \alpha ^3 m^5+233604 \alpha ^2 m^5\nonumber \\ \fl &&+318258 \alpha  m^5+120177 m^5+56736 \alpha ^4 m^4+362952 \alpha ^3 m^4+764688 \alpha ^2 m^4\nonumber \\ \fl &&+586902 \alpha  m^4+96579 m^4+39456 \alpha ^5 m^3+325584 \alpha ^4 m^3+960144 \alpha ^3 m^3\nonumber \\ \fl &&+1164216 \alpha ^2 m^3+404478 \alpha  m^3-120933 m^3+16096 \alpha ^6 m^2+166656 \alpha ^5 m^2 \nonumber \\ \fl &&+660552 \alpha ^4 m^2+1169632 \alpha ^3 m^2+731072 \alpha ^2 m^2-263998 \alpha  m^2-355132 m^2\nonumber \\ \fl &&+3392 \alpha ^7 m+44064 \alpha ^6 m+233232 \alpha ^5 m+587336 \alpha ^4 m+622960 \alpha ^3 m-48416 \alpha ^2 m\nonumber \\ \fl &&-587444 \alpha  m-315984 m-102840)+
             k^4 (1536 \alpha ^8+26752 \alpha ^7+183104 \alpha ^6+654656 \alpha ^5\nonumber \\ \fl &&+1328576 \alpha ^4+1512896 \alpha ^3+856760 \alpha ^2+134196 \alpha +6318 m^8+50544 \alpha  m^7\nonumber \\ \fl &&+75816 m^7+171504 \alpha ^2 m^6+525312 \alpha  m^6+381783 m^6+321408 \alpha ^3 m^5\nonumber \\ \fl &&+1511136 \alpha ^2 m^5+2242098 \alpha  m^5+1047843 m^5+361872 \alpha ^4 m^4+2330784 \alpha ^3 m^4\nonumber \\ \fl &&+5324988 \alpha ^2 m^4+5079846 \alpha  m^4+1686189 m^4+248256 \alpha ^5 m^3+2068128 \alpha ^4 m^3\nonumber \\ \fl &&+6521976 \alpha ^3 m^3+9617412 \alpha ^2 m^3+6517542 \alpha  m^3+1564101 m^3+99584 \alpha ^6 m^2\nonumber \\ \fl &&+1043520 \alpha ^5 m^2+4311472 \alpha ^4 m^2+8854024 \alpha ^3 m^2+9340320 \alpha ^2 m^2+4595918 \alpha  m^2\nonumber \\ \fl &&+724636 m^2+20608 \alpha ^7 m+271296 \alpha ^6 m+1437248 \alpha ^5 m+3923168 \alpha ^4 m\nonumber \\ \fl &&+5847024 \alpha ^3 m+4563856 \alpha ^2 m+1526480 \alpha  m+69690 m-41956)\nonumber \\ \fl &&+
            2 k^3 (1280 \alpha ^8+23008 \alpha ^7+161152 \alpha ^6+597576 \alpha ^5+1289880 \alpha ^4+1650074 \alpha ^3\nonumber \\ \fl &&+1206664 \alpha ^2+449834 \alpha +6696 m^8+53568 \alpha  m^7+80352 m^7+180474 \alpha ^2 m^6\nonumber \\ \fl &&+555450 \alpha  m^6+410967 m^6+332892 \alpha ^3 m^5+1582182 \alpha ^2 m^5+2402676 \alpha  m^5\nonumber \\ \fl &&+1167615 m^5+365112 \alpha ^4 m^4+2394684 \alpha ^3 m^4+5627514 \alpha ^2 m^4+5615112 \alpha  m^4\nonumber \\ \fl &&+2008008 m^4+241104 \alpha ^5 m^3+2063208 \alpha ^4 m^3+6721860 \alpha ^3 m^3+10403022 \alpha ^2 m^3\nonumber \\ \fl &&+7641144 \alpha  m^3+2127087 m^3+91888 \alpha ^6 m^2+998976 \alpha ^5 m^2+4278260 \alpha ^4 m^2\nonumber \\ \fl &&+9225572 \alpha ^3 m^2+10546601 \alpha ^2 m^2+6036963 \alpha  m^2+1341995 m^2+17888 \alpha ^7 m\nonumber \\ \fl &&+246384 \alpha ^6 m+1353400 \alpha ^5 m+3867724 \alpha ^4 m+6208378 \alpha ^3 m+5569857 \alpha ^2 m \nonumber \\ \fl &&+2557891 \alpha  m+454128 m+61758)
             +2 k^2 (2272 \alpha ^7+33888 \alpha ^6+191328 \alpha ^5\nonumber \\ \fl &&+540776 \alpha ^4+841030 \alpha ^3+727732 \alpha ^2+326608 \alpha +4374 m^8+34992 \alpha  m^7+52488 m^7\nonumber \\ \fl &&+114858 \alpha ^2 m^6+359802 \alpha  m^6+269415 m^6+199260 \alpha ^3 m^5+988038 \alpha ^2 m^5\nonumber \\ \fl &&+1547964 \alpha  m^5+771363 m^5+195216 \alpha ^4 m^4+1386732 \alpha ^3 m^4+3454704 \alpha ^2 m^4\nonumber \\ \fl &&+3609678 \alpha  m^4+1345731 m^4+105648 \alpha ^5 m^3+1044984 \alpha ^4 m^3+3766428 \alpha ^3 m^3\nonumber \\ \fl &&+6289626 \alpha ^2 m^3+4927992 \alpha  m^3+1464057 m^3+27728 \alpha ^6 m^2+400128 \alpha ^5 m^2\nonumber \\ \fl &&+2047536 \alpha ^4 m^2+5000612 \alpha ^3 m^2+6303013 \alpha ^2 m^2+3947503 \alpha  m^2+971301 m^2\nonumber \\ \fl &&+2272 \alpha ^7 m+63408 \alpha ^6 m+491184 \alpha ^5 m+1740784 \alpha ^4 m+3258206 \alpha ^3 m\nonumber \\ \fl &&+3317555 \alpha ^2 m+1730691 \alpha  m+361755 m+59106)
             -k (2816 \alpha ^8+40640 \alpha ^7\nonumber \\ \fl &&+224896 \alpha ^6+644720 \alpha ^5+1065568 \alpha ^4+1048520 \alpha ^3+600572 \alpha ^2+180960 \alpha \nonumber \\ \fl &&+1500 m^8+12000 \alpha  m^7+18000 m^7+48048 \alpha ^2 m^6+132048 \alpha  m^6+93687 m^6\nonumber \\ \fl && +120288 \alpha ^3 m^5+468720 \alpha ^2 m^5+616554 \alpha  m^5+276183 m^5+190992 \alpha ^4 m^4\nonumber \\ \fl &&+983424 \alpha ^3 m^4+1869924 \alpha ^2 m^4+1578102 \alpha  m^4+502521 m^4+186048 \alpha ^5 m^3\nonumber \\ \fl &&+1229088 \alpha ^4 m^3+3131160 \alpha ^3 m^3+3888852 \alpha ^2 m^3+2376726 \alpha  m^3+573381 m^3\nonumber \\ \fl &&+104352 \alpha ^6 m^2+871200 \alpha ^5 m^2+2875272 \alpha ^4 m^2+4841976 \alpha ^3 m^2+4422324 \alpha ^2 m^2\nonumber \\ \fl &&+2086314 \alpha  m^2+394938 m^2+29376 \alpha ^7 m+311520 \alpha ^6 m+1316112 \alpha ^5 m\nonumber \\ \fl &&+2899848 \alpha ^4 m+3626472 \alpha ^3 m+2586900 \alpha ^2 m+973332 \alpha  m+146070 m+20844)
             \nonumber \\ \fl &&-3 (512 \alpha ^8+6016 \alpha ^7+29952 \alpha ^6+82816 \alpha ^5+139488 \alpha ^4+146824 \alpha ^3+94288 \alpha ^2\nonumber \\ \fl &&+33624 \alpha +350 m^8+2800 \alpha  m^7+4200 m^7+10060 \alpha ^2 m^6+29660 \alpha  m^6\nonumber \\ \fl &&+21715 m^6+21160 \alpha ^3 m^5+92100 \alpha ^2 m^5+132630 \alpha  m^5+63135 m^5+28304 \alpha ^4 m^4\nonumber \\ \fl &&+162408 \alpha ^3 m^4+345752 \alpha ^2 m^4+324098 \alpha  m^4+112693 m^4+24416 \alpha ^5 m^3\nonumber \\ \fl &&+174256 \alpha ^4 m^3+489856 \alpha ^3 m^3+679904 \alpha ^2 m^3+466210 \alpha  m^3+126033 m^3\nonumber \\ \fl &&+13120 \alpha ^6 m^2+112608 \alpha ^5 m^2+394400 \alpha ^4 m^2+724032 \alpha ^3 m^2+736348 \alpha ^2 m^2\nonumber \\ \fl &&+393194 \alpha  m^2+85778 m^2+3968 \alpha ^7 m+40128 \alpha ^6 m+169344 \alpha ^5 m+387968 \alpha ^4 m\nonumber \\ \fl &&+522696 \alpha ^3 m+414628 \alpha ^2 m+178960 \alpha  m+32208 m+5040)),
         \end{eqnarray}
         where 
         \begin{eqnarray*}
         \fl &&  D_{c_1(k)} \nonumber\\ \fl &=& (k-1) k (-2 \alpha +k-1) (-\alpha +k-1) (k-\alpha ) (\alpha +k) (\alpha +k+1) (2 \alpha +k+1)\nonumber \\ \fl && \times (2 k^5 (2 \alpha +2 m+3)^2+4 k^4 (20 \alpha ^2+62 \alpha +21 m^2+42 \alpha  m+63 m+47)\nonumber \\ \fl &&-k^3 (32 \alpha ^4+200 \alpha ^3+152 \alpha ^2-546 \alpha +36 m^4+144 \alpha  m^3+216 m^3\nonumber \\ \fl &&+212 \alpha ^2 m^2+644 \alpha  m^2+131 m^2+136 \alpha ^3 m+628 \alpha ^2 m+250 \alpha  m-579 m-606)\nonumber \\ \fl &&-k^2 (192 \alpha ^4+1232 \alpha ^3+2328 \alpha ^2+942 \alpha +234 m^4+936 \alpha  m^3+1404 m^3\nonumber \\ \fl &&+1360 \alpha ^2 m^2+4168 \alpha  m^2+2401 m^2+848 \alpha ^3 m+3992 \alpha ^2 m+4670 \alpha  m+885 m-482)\nonumber \\ \fl &&-k (352 \alpha ^4+2400 \alpha ^3+5376 \alpha ^2+4178 \alpha +500 m^4+2000 \alpha  m^3+3000 m^3+2848 \alpha ^2 m^2\nonumber \\ \fl &&+8848 \alpha  m^2+5925 m^2+1696 \alpha ^3 m+8240 \alpha ^2 m+11394 \alpha  m+4275 m+736)\nonumber \\ \fl &&-192 \alpha ^4-1464 \alpha ^3-3656 \alpha ^2-3394 \alpha -350 m^4-1400 \alpha  m^3-2100 m^3-1940 \alpha ^2 m^2\nonumber \\ \fl &&-6140 \alpha  m^2-4355 m^2-1080 \alpha ^3 m-5500 \alpha ^2 m-8230 \alpha  m-3615 m-974)
         \end{eqnarray*}
         
         \begin{eqnarray}
            \fl &&  c_2(k)\nonumber\\ \fl &=&-\frac{1}{D_{c_2(k)}} (k+3) (k+4) (-2 \alpha +k-2 m-1) (2 \alpha +k+2 m+5)(2 k^5 (2 \alpha +2 m+3)^2\nonumber \\ \fl &&+4 k^4 (m+1) (2 \alpha +m+2)-k^3 (32 \alpha ^4+200 \alpha ^3+472 \alpha ^2+478 \alpha +36 m^4+144 \alpha  m^3\nonumber \\ \fl &&+216 m^3+212 \alpha ^2 m^2+644 \alpha  m^2+483 m^2+136 \alpha ^3 m+628 \alpha ^2 m+954 \alpha  m+477 m\nonumber \\ \fl &&+178)-k^2 (m+1) (32 \alpha ^3+136 \alpha ^2+186 \alpha +18 m^3+72 \alpha  m^2+90 m^2+88 \alpha ^2 m\nonumber \\ \fl &&+232 \alpha  m+149 m+82)+k (32 \alpha ^4+128 \alpha ^3+192 \alpha ^2+126 \alpha +4 m^4+16 \alpha  m^3+24 m^3\nonumber \\ \fl &&+48 \alpha ^2 m^2+96 \alpha  m^2+59 m^2+64 \alpha ^3 m+192 \alpha ^2 m+190 \alpha  m+69 m+32)+(m+1) \nonumber \\ \fl && \times (8 \alpha ^3+24 \alpha ^2+26 \alpha +2 m^3+8 \alpha  m^2+10 m^2+12 \alpha ^2 m+28 \alpha  m+17 m+10)),
         \end{eqnarray}
         where
         \begin{eqnarray*}
           \fl &&  D_{c_2(k)}\nonumber\\ \fl&=&   (k-1) k (-2 \alpha +k-1) (-\alpha +k-1) (k-\alpha ) (\alpha +k) (\alpha +k+1) (2 \alpha +k+1)\nonumber \\ \fl && \times (2 k^5 (2 \alpha +2 m+3)^2+4 k^4 (20 \alpha ^2+62 \alpha +21 m^2+42 \alpha  m+63 m+47)\nonumber \\ \fl &&-k^3 (32 \alpha ^4+200 \alpha ^3+152 \alpha ^2-546 \alpha +36 m^4+144 \alpha  m^3+216 m^3+212 \alpha ^2 m^2\nonumber \\ \fl &&+644 \alpha  m^2+131 m^2+136 \alpha ^3 m+628 \alpha ^2 m+250 \alpha  m-579 m-606)-k^2 (192 \alpha ^4\nonumber \\ \fl &&+1232 \alpha ^3+2328 \alpha ^2+942 \alpha +234 m^4+936 \alpha  m^3+1404 m^3+1360 \alpha ^2 m^2+4168 \alpha  m^2\nonumber \\ \fl &&+2401 m^2+848 \alpha ^3 m+3992 \alpha ^2 m+4670 \alpha  m+885 m-482)-k(352 \alpha ^4+2400 \alpha ^3\nonumber \\ \fl &&+5376 \alpha ^2+4178 \alpha +500 m^4+2000 \alpha  m^3+3000 m^3+2848 \alpha ^2 m^2+8848 \alpha  m^2\nonumber \\ \fl &&+5925 m^2+1696 \alpha ^3 m+8240 \alpha ^2 m+11394 \alpha  m+4275 m+736)-192 \alpha ^4-1464 \alpha ^3\nonumber \\ \fl &&-3656 \alpha ^2-3394 \alpha -350 m^4-1400 \alpha  m^3-2100 m^3-1940 \alpha ^2 m^2-6140 \alpha  m^2\nonumber \\ \fl &&-4355 m^2-1080 \alpha ^3 m-5500 \alpha ^2 m-8230 \alpha  m-3615 m-974)
         \end{eqnarray*}
\section*{References}
\bibliographystyle{iopart-num}
\bibliography{sn-bibliography}

\end{document}